\documentclass[envcountsame,envcountsect,runningheads]{llncs}
\newcommand{\LongVersion}[1]{#1}
\newcommand{\ShortVersion}[1]{}
\newcommand{\VSpace}[1]{}

\usepackage{latexsym}
\usepackage{amssymb}
\usepackage{stmaryrd}
\usepackage{url}
\input{xy}
\xyoption{all}
\ShortVersion{\usepackage{times}}

\newcommand{\V}{\forall}
\newcommand{\E}{\exists}
\newcommand{\Dmd}{\Diamond}
\newcommand{\lDmd}[1]{\langle #1 \rangle}

\newcommand{\rBotz}{(\bot_0)}
\newcommand{\rBot}{(\bot)}
\newcommand{\rAnd}{(\land)}
\newcommand{\rOr}{(\lor)}
\newcommand{\rBoxSc}{(\Box_{;})}
\newcommand{\rDmdSc}{(\Dmd_{;})}
\newcommand{\rBoxQm}{(\Box_{?})}
\newcommand{\rDmdQm}{(\Dmd_{?})}
\newcommand{\rBoxCp}{(\Box_{\cup})}
\newcommand{\rDmdCp}{(\Dmd_{\cup})}
\newcommand{\rBoxSt}{(\Box_{*})}
\newcommand{\rDmdSt}{(\Dmd_{*})}
\newcommand{\rTrans}{(trans)}

\newcommand{\rBotzP}{(\bot'_0)}
\newcommand{\rBotP}{(\bot')}
\newcommand{\rAndP}{(\land')}
\newcommand{\rOrP}{(\lor')}
\newcommand{\rBoxScP}{(\Box'_{;})}
\newcommand{\rDmdScP}{(\Dmd'_{;})}
\newcommand{\rBoxQmP}{(\Box'_{?})}
\newcommand{\rDmdQmP}{(\Dmd'_{?})}
\newcommand{\rBoxCpP}{(\Box'_{\cup})}
\newcommand{\rDmdCpP}{(\Dmd'_{\cup})}
\newcommand{\rBoxStP}{(\Box'_{*})}
\newcommand{\rDmdStP}{(\Dmd'_{*})}
\newcommand{\rBoxP}{(\Box')}
\newcommand{\rTransP}{(trans')}

\newcommand{\cPDL}{$\mathcal{C}_{\mathrm{PDL}}$}
\newcommand{\cPDLA}{$\mathcal{C}_{\mathrm{PDL+ABox}}$}

\newcommand{\mM}{\mathcal{M}}
\newcommand{\mL}{\mathcal{L}}
\newcommand{\mG}{\mathcal{G}}

\newcommand{\mA}{\mathcal{A}}

\newcommand{\props}{\Phi_0}
\newcommand{\mindices}{\Pi_0}
\newcommand{\rfs}{\mathit{rfs}} 

\newcommand{\fracc}[2]{\displaystyle{\frac{\;#1\;}{\;#2\;}}}

\newtheorem{Algorithm}{Algorithm}
\newenvironment{algorithm}{\begin{Algorithm}\begin{em}}{\end{em}\end{Algorithm}}

\newcommand{\comment}[1]{}

\newcommand{\Unexpanded}{\textsf{unexpanded}}
\newcommand{\Expanded}{\textsf{expanded}}
\newcommand{\Unsat}{\textsf{unsat}}
\newcommand{\Sat}{\textsf{sat}}

\newcommand{\EXPTIME}{{\footnotesize\textsc{ExpTime}}}
\newcommand{\NEXPTIME}{{\footnotesize\textsc{NExpTime}}}

\newcommand{\ALC}{\mathcal{ALC}}

\def\ramka#1{\begin{center}
\fbox{\parbox{11cm}{\begin{center}#1\end{center}}}
\end{center}}

\def\trojkat{\mbox{{\scriptsize$\!\vartriangleleft$}}}
\newcommand{\koniec}{\mbox{}\hfill\trojkat}

\LongVersion{
\title{Optimal Tableau Decision Procedures for {PDL}}
\titlerunning{Optimal Tableau Decision Procedures for {PDL}}
}

\ShortVersion{
\title{\VSpace{-3em}Checking Consistency of an ABox\\ w.r.t. Global Assumptions in PDL\thanks{Supported by the MNiSW grant N~N206~399334.}}
\pagestyle{plain}
}

\author{Linh Anh Nguyen \and Andrzej Sza\l{}as}
\institute{
Institute of Informatics, University of Warsaw\\
Banacha 2, 02-097 Warsaw, Poland\\
\email{\{nguyen,andsz\}@mimuw.edu.pl}\\
\LongVersion{
\bigskip
February 2, 2009 (last revised: \today)
}
}
\authorrunning{L.A. Nguyen and A. Sza{\l}as}

\begin{document}
\maketitle \sloppy

\begin{abstract}\VSpace{-2em}
We reformulate Pratt's tableau decision procedure of checking satisfiability of a~set of formulas in PDL. Our formulation is simpler and more direct for implementation. Extending the method we give the first \EXPTIME\ (optimal) tableau decision procedure not based on transformation for checking consistency of an ABox w.r.t.\ a~TBox in PDL (here, PDL is treated as a~description logic). We also prove the new result that the data complexity of the instance checking problem in PDL is coNP-complete.
\end{abstract}

\VSpace{-2em}
\section{Introduction}\VSpace{-0.5em}

Propositional dynamic logic (PDL) is a~multi-modal logic introduced by Fischer and Ladner \cite{FischerLadner79} for reasoning about programs. It is useful not only for program verification but also for other fields of computer science like knowledge representation and artificial intelligence (e.g., \cite{HalpernMoses92,HKT00,handbook-AR-2001,ddks2007}). For example, the description logic $\ALC_{reg}$, a~notational variant of PDL, can be used for reasoning about structured knowledge.

The problem of checking satisfiability of a~set of formulas in PDL is \EXPTIME-complete. This result was established by Fischer and Ladner~\cite{FischerLadner79}, but their decision procedure for PDL is via filtration and canonical model and therefore is not really practical. The first practical and optimal (\EXPTIME) decision procedure for PDL was given by Pratt~\cite{Pratt80}. The essence of his procedure is based on constructing an ``and-or'' graph for the considered set of formulas by using tableau rules and global caching, and then checking whether a~model for the set can be extracted from the graph. However, the formulation of his procedure is a~bit too indirect: it goes via a~labeled tableau calculus, tree-like labeled tableaux, tree-like traditional (``lean'') tableaux, and ``and-or'' graphs.

De Giacomo and Massacci~\cite{de-giacomo-massacci-converse-pdl} gave a~\NEXPTIME\ algorithm for checking satisfiability in CPDL (i.e., PDL with converse) and \ShortVersion{informally} described how to transform the algorithm to an \EXPTIME\ version. 
\LongVersion{However, the description is informal and unclear: the transformation is based on Pratt's global caching method formulated for PDL~\cite{Pratt80}, but no global caching method has been formalized and proved sound for labeled tableaux that allow modifying labels of ancestor nodes in order to deal with converse.\footnote{Gor{\'e} and Nguyen have recently formalized sound global caching
\LongVersion{\cite{GoreNguyen05tab,GoreNguyenDL07,GoreNguyenTab07,GoreNguyen07clima,GoreNguyen08CSP}}
\ShortVersion{\cite{GoreNguyenTab07,GoreNguyen08CSP}}
for
\LongVersion{{\em traditional} (unlabeled)}
\ShortVersion{{\em unlabeled}}
tableaux in a~number of modal logics without the $*$ operator, which {\em never modify ancestor nodes}.}
} 
Abate et al.~\cite{AbateGW07} gave a~``single-pass'' tableau decision procedure for checking satisfiability in PDL.
Their algorithm does not exploit global caching~\cite{Pratt80,GoreNguyen08CSP} and has complexity 2\EXPTIME\ in the
worst cases. There are a~few prototype implementations for checking satisfiability in
PDL~\cite{SchmidtPDLtab,LoTREC,AbateGW07}.

There is a~tight relationship between multi-modal logics and description logics which will often be exploited in this paper. Two basic components of description logic theories are ABoxes and TBoxes. An ABox ({\em assertion box}) consists of facts and a~TBox ({\em terminological box}) consists of formulas expressing relationships between concepts. Two basic reasoning problems considered in description logics, amongst others, are:\VSpace{-0.5em}
 \begin{enumerate}
 \item the problem of checking consistency of an ABox w.r.t.\ a~TBox,
 \item the instance checking problem.
 \end{enumerate}\VSpace{-0.5em}


The first tableau-based procedure for $\ALC_{reg}$ (PDL) in the description logic context was proposed by Baader~\cite{ijcai-Baader91} (the correspondence between $\ALC_{reg}$ and PDL had not yet been known). His procedure, however, has non-optimal complexity 2\EXPTIME. The correspondence between description logics like $\ALC_{reg}$ and PDL was first described in Schild's paper~\cite{Schild91}. 
In~\cite{DeGiacomoThesis}, encoding the ABox by ``nominals'' and ``internalizing'' the TBox, De Giacomo showed that the complexity of checking consistency of an ABox w.r.t.\ a~TBox in CPDL is \EXPTIME-complete. In~\cite{GiacomoL96}, using a transformation that encodes the ABox by a concept assertion plus terminology axioms, De Giacomo and Lenzerini showed that the mentioned problem is also \EXPTIME-complete for the description logic $\mathcal{CIQ}$ (an extension of CPDL). 


In this paper, we reformulate Pratt's algorithm of checking satisfiability of a~set of formulas in PDL. Our formulation is directly based on building an ``and-or'' graph by using traditional (unlabeled) tableau rules and global caching and is therefore simpler and more direct for implementation. Extending the method we give the {\em first} \EXPTIME\ (optimal) tableau decision procedure not based on transformation (encoding) for checking consistency of an ABox w.r.t.\ a~TBox in~PDL. 

Despite that the upper-bound \EXPTIME\ is known for the complexity of the mentioned satisfiability problem in CPDL, implemented tableau provers for description logics usually have non-optimal complexity 2\EXPTIME.
\LongVersion{In the well-known overview~\cite{BaaderSattler01},}
\ShortVersion{In the overview~\cite{BaaderSattler01},}
Baader and Sattler wrote: {\em ``The point in designing these [non-optimal] algorithms was not to prove worst-case complexity results, but \ldots\ to obtain `practical' algorithms \ldots\ that are easy to implement and optimise, and which behave well on realistic knowledge bases. Nevertheless, the fact that `natural' tableau algorithms for such \EXPTIME-complete logics are usually \NEXPTIME-algorithms is an unpleasant phenomenon. \ldots\ Attempts to design \EXPTIME-tableaux for such logics (De Giacomo et al., 1996; De Giacomo and Massacci, 1996; Donini and Massacci, 1999) usually lead to rather complicated (and thus not easy to implement) algorithms, which (to the best of our knowledge) have not been implemented yet.''}~\cite[page 26]{BaaderSattler01}. 

Our formulation of tableau calculi and decision procedures for PDL is short and clear, which makes the procedures natural and easy to implement. The first author has implemented a tableau prover called TGC for the basic description logic $\ALC$, which is also based on ``and-or'' graphs with global caching. The test results of TGC on the sets T98-sat and T98-kb of DL'98 Systems Comparison are comparable with the test results of the best systems DLP-98 and FaCT-98 that took part in that comparison (see \cite{Nguyen08CSP-FI}). One can say that the mentioned test sets are not representative for practical applications, but the comparison at least shows that optimization techniques can be applied (not only for $\ALC$ but also PDL) to obtain decision procedures that are both efficient in practice and optimal w.r.t.\ complexity. 

We also study the data complexity of the instance checking problem in PDL. For the well-known description logic $\mathcal{SHIQ}$, Hustadt et al.~\cite{HustadtMS05} proved that the data complexity of that problem is coNP-complete. The lower bound for the data complexity of that problem in PDL ($\ALC_{reg}$) is known to be coNP-hard (shown for $\ALC$ by Schaerf in~\cite{Schaerf94}). In this paper, by establishing the upper bound, we prove the new result that the data complexity of the instance checking problem in PDL is coNP-complete. 

The rest of this paper is structured as follows. In Section~\ref{section: defs PDL}, we define syntax and semantics of PDL. In Section~\ref{section:problems} we formulate the problems we deal with. In Section~\ref{section: cal1}, we present a~tableau calculus for checking satisfiability of a~set of formulas w.r.t.\ a~set of global assumptions in PDL. In Section~\ref{section: ABoxes}, we extend that calculus for checking consistency of an ABox w.r.t.\ a~set of global assumptions (i.e.,\ a~TBox) in PDL. In Section~\ref{section: dp-cr}, we give decision procedures based on our tableau calculi for the mentioned problems and derive the data complexity result. \LongVersion{In Section~\ref{section: opts}, we discuss optimizations for our decision procedures.} Conclusions are given in Section~\ref{section: conc}. 
\LongVersion{Proofs of soundness and completeness of our calculi are presented in the appendices.}
\ShortVersion{Due to the lack of space, proofs of soundness and completeness of the calculi are contained in the full version~\cite{pdl-tab-long} of this~paper.}

\VSpace{-0.5em}
\section{Propositional Dynamic Logic}\VSpace{-0.5em}
\label{section: defs PDL}

We use $\mindices$ to denote the set of {\em atomic programs}, and $\props$ to denote the set of {\em propositions}
(i.e.,\ atomic formulas). We denote elements of $\mindices$ by letters like $\sigma$, and elements of $\props$ by
letters like $p$, $q$. {\em Formulas} and {\em programs} of PDL are defined respectively by the following BNF grammar
rules:
\[
\begin{array}{rcl}
\varphi & ::= &
    \top
    \mid \bot
    \mid p
    \mid \lnot \varphi
    \mid \varphi \land \varphi
    \mid \varphi \lor \varphi
    \mid \varphi \to \varphi
    \mid \lDmd{\alpha}\varphi
    \mid [\alpha]\varphi \\[0.5ex]
\alpha & ::= &
    \sigma
    \mid \alpha;\alpha
    \mid \alpha \cup \alpha
    \mid \alpha^*
    \mid \varphi?
\end{array}
\]
We use letters like $\alpha$, $\beta$ to denote programs, and $\varphi$, $\psi$, $\xi$ to denote formulas.

A {\em Kripke model} is a~pair $\mM = \langle \Delta^\mM, \cdot^\mM\rangle$, where $\Delta^\mM$ is a~set of {\em
states}, and $\cdot^\mM$ is an interpretation function that maps each proposition $p$ to a~subset $p^\mM$ of
$\Delta^\mM$, and each atomic program $\sigma$ to a~binary relation $\sigma^\mM$ on $\Delta^\mM$. The interpretation
function is extended to interpret complex formulas and complex programs as follows:\VSpace{-0.5em}
\[
\begin{array}{l}
\top^\mM   =   \Delta^\mM,\ \ \bot^\mM  =  \emptyset,\ \
(\lnot\varphi)^\mM  =  \Delta^\mM \setminus \varphi^\mM\\
(\varphi \land \psi)^\mM   =   \varphi^\mM \cap \psi^\mM,\ \ (\varphi \lor \psi)^\mM   =   \varphi^\mM \cup \psi^\mM, \
\
(\varphi \to \psi)^\mM   =   (\lnot\varphi \lor \psi)^\mM \\
(\lDmd{\alpha}\varphi)^\mM   =   \{ x \in \Delta^\mM \mid
    \E y[\alpha^\mM(x,y) \land \varphi^\mM(y)] \} \\
([\alpha]\varphi)^\mM   =   \{ x \in \Delta^\mM \mid
    \V y[ \alpha^\mM(x,y) \to \varphi^\mM(y)] \} \\[1.5ex]

(\alpha;\beta)^\mM   =   \alpha^\mM \circ \beta^\mM = \{ (x,y) \mid \E z[\alpha^\mM(x,z) \land \beta^\mM(z,y)] \}\\
(\alpha \cup \beta)^\mM   =   \alpha^\mM \cup \beta^\mM, \ \ (\alpha^*)^\mM   =   (\alpha^\mM)^*, \ \ (\varphi?)^\mM
=   \{ (x,x) \mid \varphi^\mM(x) \}\VSpace{-0.5em}
\end{array}
\]

We write $\mM,w \models \varphi$ to denote $w \in \varphi^\mM$. For a~set $X$ of formulas, we write $\mM,w \models X$
to denote that $\mM,w \models \varphi$ for all $\varphi \in X$. If $\mM,w \models \varphi$ (resp.\ $\mM,w \models X$),
then we say that $\mM$ {\em satisfies} $\varphi$ (resp.\ $X$) {\em at} $w$, and that $\varphi$ (resp.\ $X$) is {\em
satisfied at} $w$ in $\mM$. We say that $\mM$ {\em validates} $X$ if $\mM,w \models X$ for all $w \in \Delta^\mM$, and
that $X$ is {\em satisfiable} w.r.t.\ a~set $\Gamma$ of formulas used as {\em global assumptions} if there exists a~
Kripke model that validates $\Gamma$ and satisfies $X$ at some state.

\ShortVersion{
The Fischer-Ladner closure $FL(X)$ for a set $X$ of formulas in negation normal form (NNF) is defined as usual (see~\cite{PDL-long,HKT00}).\footnote{In NNF, the connective $\to$ does not occur and $\lnot$ occurs only immediately before propositions. Every formula can be transformed to an equivalent formula in NNF.}  
}
\LongVersion{
The Fischer-Ladner closure $FL(\varphi)$ and the sets $FL^\Box([\alpha]\varphi)$ and $FL^\Dmd(\lDmd{\alpha}\varphi)$, where $\varphi$ is a formula in negation normal form (NNF), are the sets of formulas defined as follows:\footnote{In NNF, the connective $\to$ does not occur and $\lnot$ occurs only immediately before propositions. Every formula can be transformed to an equivalent formula in NNF.}
\[
\begin{array}{l}
FL(\top)  =  \{\top\},\;\;
FL(\bot)  =  \{\bot\},\;\;
FL(p)  =  \{p\}, \;\;
FL(\lnot p)  =  \{\lnot p\}, \\
FL(\varphi \land \psi)   =   \{\varphi \land \psi\} \cup FL(\varphi) \cup FL(\psi), \\
FL(\varphi \lor \psi)   =   \{\varphi \lor \psi\} \cup FL(\varphi) \cup FL(\psi), \\
FL([\alpha]\varphi)   =   FL^\Box([\alpha]\varphi) \cup FL(\varphi), \;\;
FL(\lDmd{\alpha}\varphi)   =   FL^\Dmd(\lDmd{\alpha}\varphi) \cup FL(\varphi), \\[1ex]
FL^\Box([\sigma]\varphi)   =   \{[\sigma]\varphi\}, \\
FL^\Box([\alpha;\beta]\varphi)   =   \{[\alpha;\beta]\varphi\} \cup FL^\Box([\alpha][\beta]\varphi) \cup FL^\Box([\beta])\varphi, \\
FL^\Box([\alpha \cup \beta]\varphi)   =   \{[\alpha \cup \beta]\varphi\} \cup FL^\Box([\alpha]\varphi) \cup FL^\Box([\beta])\varphi, \\
FL^\Box([\alpha^*]\varphi)   =   \{[\alpha^*]\varphi\} \cup FL^\Box([\alpha][\alpha^*]\varphi), \\
FL^\Box([\psi?]\varphi)   =   \{[\psi?]\varphi\} \cup FL(\overline{\psi}), \\[1ex]
FL^\Dmd(\lDmd{\sigma}\varphi)   =   \{\lDmd{\sigma}\varphi\}, \\
FL^\Dmd(\lDmd{\alpha;\beta}\varphi)   =   \{\lDmd{\alpha;\beta}\varphi\} \cup FL^\Dmd(\lDmd{\alpha}\lDmd{\beta}\varphi) \cup FL^\Dmd(\lDmd{\beta})\varphi, \\
FL^\Dmd(\lDmd{\alpha \cup \beta}\varphi)   =   \{\lDmd{\alpha \cup \beta}\varphi\} \cup FL^\Dmd(\lDmd{\alpha}\varphi) \cup FL^\Dmd(\lDmd{\beta})\varphi, \\
FL^\Dmd(\lDmd{\alpha^*}\varphi)   =   \{\lDmd{\alpha^*}\varphi\} \cup FL^\Dmd(\lDmd{\alpha}\lDmd{\alpha^*}\varphi), \\
FL^\Dmd(\lDmd{\psi?}\varphi)   =   \{\lDmd{\psi?}\varphi\} \cup FL(\psi).
\end{array}
\]

For a~set $X$ of formulas in NNF, define $FL(X) = \bigcup_{\varphi \in X} FL(\varphi)$. 
}

\VSpace{-0.5em}
\section{The Problems We Address}\label{section:problems}\VSpace{-0.5em}

When interpreting PDL as a~description logic, states in a~Kripke model, formulas, and programs are regarded
respectively as ``objects'', ``concepts'', and ``roles''. A finite set $\Gamma$ of global assumptions is treated as
a~``TBox''. As for description logics, we introduce ABoxes and consider the problem of checking whether a~given ABox is
consistent with a~given TBox, which is related to the instance checking problem.

We prefer to use the terminology of PDL instead of that of $\ALC_{reg}$ because this work is related to Pratt's work on PDL.
We use the term {\em state variable} as an equivalent for the term ``individual'' used in description logic, and use
letters like $a$, $b$, $c$ to denote state variables. We extend the notion of Kripke model so that the interpretation
function $\cdot^\mM$ of a~Kripke model $\mM$ maps each state variable $a$ to a~state $a^\mM$ of $\mM$.

An {\em ABox} is a~finite set of {\em assertions} of the form $a\!:\!\varphi$ or $\sigma(a,b)$, where $\varphi$ is a formula in NNF and $a$ is a state variable. The meaning of $a\!:\!\varphi$ is that formula $\varphi$ is satisfied in state $a$. An ABox is {\em extensionally reduced} if it contains only assertions of the form $a\!:\!p$ or $\sigma(a,b)$.
We will refer to ABox assertions also as formulas. When necessary, we refer to formulas that are not ABox assertions as {\em traditional formulas}.

A {\em TBox} is a~finite set of traditional formulas in NNF.

A Kripke model $\mM$ {\em satisfies} an ABox $\mA$ if $a^\mM \in \varphi^\mM$ for all $(a\!:\!\varphi) \in \mA$ and
$(a^\mM,b^\mM) \in \sigma^\mM$ for all $\sigma(a,b) \in \mA$. An ABox $\mA$ is {\em satisfiable w.r.t.} (or {\em
consistent with}) a~TBox $\Gamma$ iff there exists a~Kripke model $\mM$ that satisfies $\mA$ and validates~$\Gamma$.

The first problem we address is the problem of checking satisfiability of an ABox w.r.t.\ a~TBox

Consider the use of PDL as a~description logic. A pair $(\mA,\Gamma)$ of an ABox $\mA$ and a~TBox $\Gamma$ is treated
as a~knowledge base. A Kripke model that satisfies $\mA$ and validates $\Gamma$ is called a~model of $(\mA,\Gamma)$.
Given a~(traditional) formula $\varphi$ (treated as a~``concept'') and a~state variable $a$ (treated as an
``individual''), the problem of checking whether $a^\mM \in \varphi^\mM$ in every model $\mM$ of $(\mA,\Gamma)$ is called
the instance checking problem (in PDL).

The second problem considered in this paper is the instance checking problem. The condition to check is denoted in such
cases by $(\mA,\Gamma) \models \varphi(a)$.

\VSpace{-0.5em}
\section{A Tableau Calculus for PDL}\VSpace{-0.5em}
\label{section: cal1}

In this section, we do not consider ABoxes yet, and by a ``formula'' we mean a ``traditional formula''. Let $X$ and
$\Gamma$ be finite sets of formulas.
Consider the problem of checking whether $X$ is satisfiable in PDL w.r.t.\ the set $\Gamma$ of global assumptions. We
assume that formulas are represented in NNF. We write $\overline{\varphi}$ to denote the NNF of~$\lnot\varphi$.


We will define tableaux as ``and-or'' graphs. The {\em contents} of a {\em node} $v$ of an ``and-or'' graph are a data structure consisting of two sets $\mL(v)$ and $\rfs(v)$ of formulas, where $\mL(v)$ is called the {\em label} of $v$, and $\rfs(v)$ is called ``the set of formulas that have been reduced by a static rule after the last application of the transitional rule''. 

Our calculus \cPDL\ will be specified as a finite set of tableau rules, which are used to expand nodes of ``and-or'' graphs. A {\em tableau rule} is specified with the following informations: 
\begin{itemize}
\item the kind of the rule: an {\em``and''-rule} or an {\em``or''-rule}, 
\item the conditions for applicability of the rule (if any),
\item the priority of the rule, 
\item the number of successors of a node resulting from applying the rule to it, and the way to compute their contents.
\end{itemize}

Usually, a tableau rule is written downwards, with a set of formulas above the line as the {\em premise}, which represents the label of the node to which the rule is applied, and a number of sets of formulas below the line as the {\em (possible) conclusions}, which represent the labels of the successor nodes resulting from the application of the rule.\LongVersion{\footnote{In~\cite{Gore99,GoreNguyenTab07}, ``premise'' and ``possible conclusion'' are called {\em numerator} and {\em denominator}, respectively.}} Possible conclusions of an ``or''-rule are separated by $\mid$, while conclusions of an ``and''-rule are separated/specified using~$\&$. If a~rule is a~unary rule (i.e.\ a rule with only one possible conclusion) or an ``and''-rule then its conclusions are ``firm'' and we ignore the word ``possible''.
An ``or''-rule has the meaning that, if the premise is satisfiable w.r.t.\ $\Gamma$ then some of the possible conclusions is also satisfiable w.r.t.\ $\Gamma$. On the other hand, an ``and''-rule has the meaning that, if the premise is satisfiable w.r.t.\ $\Gamma$ then all of the conclusions are also satisfiable w.r.t.\ $\Gamma$ (possibly in different states of the model under construction).
Note that, apart from the labels, there are also sets $\rfs(\_)$ to be specified for the successor nodes. 

\begin{table}[t]
\ramka{
\(
\begin{array}{ll}
\rBotz\; \fracc{Y, \bot}{\bot} \hspace{16.5ex} &
\rBot\; \fracc{Y, p, \lnot p}{\bot} \\
\ \\
\rAnd\; \fracc{Y, \varphi \land \psi}{Y, \varphi, \psi} &
\rOr\; \fracc{Y, \varphi \lor \psi}{Y, \varphi \mid Y, \psi} \\
\ \\
\rBoxSc\; \fracc{Y, [\alpha;\beta]\varphi}{Y, [\alpha][\beta]\varphi} &
\rDmdSc\; \fracc{Y, \lDmd{\alpha;\beta}\varphi}{Y, \lDmd{\alpha}\lDmd{\beta}\varphi} \\
\ \\
\rBoxCp\; \fracc{Y, [\alpha \cup \beta]\varphi}{Y, [\alpha]\varphi, [\beta]\varphi} &
\rDmdCp\; \fracc{Y, \lDmd{\alpha \cup \beta}\varphi}{Y, \lDmd{\alpha}\varphi \mid Y, \lDmd{\beta}\varphi} \\
\ \\
\rBoxQm\; \fracc{Y, [\psi?]\varphi}{Y, \overline{\psi} \mid Y, \varphi} &
\rDmdQm\; \fracc{Y, \lDmd{\psi?}\varphi}{Y, \psi, \varphi} \\
\ \\
\rBoxSt\; \fracc{Y, [\alpha^*]\varphi}{Y, \varphi, [\alpha][\alpha^*]\varphi} &
\rDmdSt\; \fracc{Y, \lDmd{\alpha^*}\varphi}{Y, \varphi \mid Y, \lDmd{\alpha}\lDmd{\alpha^*}\varphi} \\
\ \\
\multicolumn{2}{l}{\rTrans\; \fracc{Y}{\& \{\; (\{\varphi\} \cup \{\psi \textrm{ s.t. } [\sigma]\psi \in Y\} \cup \Gamma) \textrm{ s.t. } \lDmd{\sigma}\varphi \in Y \;\}}}
\end{array}
\)
} \caption{Rules of the tableau calculus \cPDL}\VSpace{-2.5em} \label{table: cPDL}
\end{table}

We use $Y$ to denote a~set of formulas, and write $Y,\varphi$ for $Y \cup \{\varphi\}$. 

Define tableau calculus \cPDL\ w.r.t.\ a~set $\Gamma$ of global assumptions to be the set of the tableau rules given in
Table~\ref{table: cPDL}. The rule $\rTrans$ is the only ``and''-rule and the only {\em transitional rule}.
Instantiating this rule, for example, to $Y = \{\lDmd{\sigma}p,\lDmd{\sigma}q,[\sigma]r\}$ and $\Gamma = \{s\}$ we get
two conclusions: $\{p,r,s\}$ and $\{q,r,s\}$. The other rules of \cPDL\ are ``or''-rules, which are also called {\em
static rules}.\LongVersion{\footnote{Unary static rules can be treated either as ``and''-rules or as ``or''-rules.
In~\cite{GoreNguyen08CSP}, the rules $\rBotz$ and $\rBot$ are classified as {\em terminal rules}.}} The intuition of the
sorting of static/transitional is that the static rules keep us in the same state of the model under construction,
while each conclusion of the transitional rule takes us to a~new state. 
For any rule of \cPDL\ except $\rTrans$, the distinguished formulas of the premise are called the {\em principal formulas} of the rule. The principal formulas of the rule $\rTrans$ are the formulas of the form $\lDmd{\sigma}\varphi$ of the premise.
We assume that any one of the rules $\rAnd$, $\rOr$, $\rBoxSc$, $\rBoxCp$, $\rBoxQm$, $\rBoxSt$ is applicable to a node $v$ only when the principal formula does not belong to $\rfs(v)$. 
Applying a static rule different from $\rBotz$ and $\rBot$ to a node $v$, for any successor node $w$ of $v$, let $\rfs(w)$ be the set that extends $\rfs(v)$ with the principal formula of the applied rule. 
Applying any other rule to a node $v$, for any successor node $w$ of $v$, let $\rfs(w) = \emptyset$.

Observe that, by using $\rfs(\_)$ and the restriction on applicability of the rules $\rAnd$, $\rOr$, $\rBoxSc$, $\rBoxCp$, $\rBoxQm$, and $\rBoxSt$, in any sequence of applications of static rules a formula of the form $\varphi \land \psi$, $\varphi \lor \psi$, $[\alpha;\beta]\varphi$, $[\alpha\cup\beta]\varphi$, $[\psi?]\varphi$, or $[\alpha^*]\varphi$ is reduced (as a principal formula) at most once. We do not adopt such a restriction for the rules $\rDmdSc$, $\rDmdCp$, $\rDmdQm$, and $\rDmdSt$ because we will require formulas of the form $\lDmd{\alpha}\varphi$ to be ``realized'' (in a finite number of steps).

We assume the following preferences for the rules of \cPDL: the rules $\rBotz$ and $\rBot$ have the highest priority;
unary static rules have a~higher priority than non-unary static rules; all the static rules have a~higher priority than
the transitional rule $\rTrans$.

An {\em ``and-or'' graph for $(X,\Gamma)$}, also called a~{\em tableau for $(X,\Gamma)$}, is an ``and-or'' graph defined as
follows.
The initial node $\nu$ of the graph, called the {\em root} of the graph, is specified by $\mL(\nu) = X \cup \Gamma$ and $\rfs(\nu) = \emptyset$. 
For every node $v$ of the graph, if a tableau rule of \cPDL\ is applicable to the label of $v$ in the sense that an instance of the rule has $\mL(v)$ as the premise and $Z_1$, \ldots, $Z_k$ as the possible conclusions, then choose such a~rule accordingly to the preference\footnote{If there are several applicable rules with the same priority, choose any one of them.} and apply it to $v$ to create $k$ successors $w_1,\ldots,w_k$ of $v$ with $\mL(w_i) = Z_i$ for $1 \leq i \leq k$. If the graph already contains a node $w'_i$ with the same contents as $w_i$ then instead of creating a new node $w_i$ as a~successor of $v$ we just connect $v$ to $w'_i$ and assume $w_i = w'_i$. If the applied rule is $\rTrans$ then we {\em label} the edge $(v,w_i)$ by the principal formula corresponding to the successor $w_i$. If the rule expanding $v$ is an ``or''-rule then $v$ is an {\em ``or''-node}, else $v$ is an {\em ``and''-node}. The information about which rule is applied to $v$ is recorded for later uses. If no rule is applicable to $v$ then $v$ is an {\em end node}. Note that each node is ``expanded'' only once (using one rule). Also note that the graph is constructed using {\em global caching} \cite{Pratt80,GoreNguyenTab07,GoreNguyen08CSP} and the contents of its nodes are unique.

A {\em marking} of an ``and-or'' graph $G$ is a~subgraph $G'$ of $G$ such that:\VSpace{-0.7em}
\begin{itemize}
\item the root of $G$ is the root of $G'$.
\item if $v$ is a~node of $G'$ and is an ``or''-node of $G$ then there exists at least one edge $(v,w)$ of $G$ that is an edge of $G'$.
\item if $v$ is a~node of $G'$ and is an ``and''-node of $G$ then every edge $(v,w)$ of $G$ is an edge of $G'$.
\item if $(v,w)$ is an edge of $G'$ then $v$ and $w$ are nodes of $G'$.
\end{itemize}

Let $G$ be an ``and-or'' graph for $(X,\Gamma)$, $G'$ a~marking of $G$, $v$ a~node of $G'$, and $\lDmd{\alpha}\varphi$ a~formula of the label of $v$. A {\em trace} of $\lDmd{\alpha}\varphi$ in $G'$ starting from $v$ is a~sequence $(v_0,\varphi_0)$, \ldots, $(v_k,\varphi_k)$ such that:\footnote{This definition of trace is inspired by~\cite{NiwinskiW96}.}
\begin{itemize}
\item \VSpace{-0.5em}$v_0 = v$ and $\varphi_0 = \lDmd{\alpha}\varphi$;
\item for every $1 \leq i \leq k$, $(v_{i-1},v_i)$ is an edge of $G'$;
\item for every $1 \leq i \leq k$, $\varphi_i$ is a~formula of the label of $v_i$ such that: if $\varphi_{i-1}$ is not a~principal formula of the tableau rule expanding $v_{i-1}$, then the rule must be a~static rule and $\varphi_i = \varphi_{i-1}$, else
  \begin{itemize}
  \item if the rule is $\rDmdSc$, $\rDmdCp$ or $\rDmdSt$ then $\varphi_i$ is the formula obtained from~$\varphi_{i-1}$,
  \item if the rule is $\rDmdQm$ and $\varphi_{i-1} = \lDmd{\psi?}\xi$ then $\varphi_i = \xi$,
  \item else the rule is $\rTrans$, $\varphi_{i-1}$ is of the form $\lDmd{\sigma}\xi$ and is the label of the edge $(v_{i-1},v_i)$, and $\varphi_i = \xi$.
  \end{itemize}
\end{itemize}\VSpace{-1em}
A trace $(v_0,\varphi_0)$, \ldots, $(v_k,\varphi_k)$ of $\lDmd{\alpha}\varphi$ in $G'$ is called a~{\em $\Dmd$-realization} in $G'$ for $\lDmd{\alpha}\varphi$ at $v_0$ if $\varphi_k = \varphi$.

A marking $G'$ of an ``and-or'' graph $G$ for $(X,\Gamma)$ is {\em consistent} if:
\begin{description}
\item[local consistency:] \VSpace{-0.5em}$G'$ does not contain any node with label $\{\bot\}$;
\item[global consistency:] for every node $v$ of $G'$, every formula of the form $\lDmd{\alpha}\varphi$ of the label of $v$ has a~$\Dmd$-realization (starting at $v$) in $G'$.
\end{description}

\VSpace{-0.5em}
\begin{theorem}[Soundness and Completeness of \cPDL]
\label{theorem: s-c-PDL} Let $X$ and $\Gamma$ be finite sets of formulas in NNF, and $G$ be an ``and-or'' graph for
$(X,\Gamma)$. Then $X$ is satisfiable w.r.t.\ the set $\Gamma$ of global assumptions iff $G$ has a~consistent marking. \koniec
\end{theorem}\VSpace{-0.5em}

\LongVersion{The ``only if'' direction means soundness of \cPDL, while the ``if'' direction means completeness of \cPDL. See Appendix~\ref{section: proof PDL} for the proof of this theorem.} 

\ShortVersion{See the full version \cite{pdl-tab-long} for the proof of this theorem.}

\newcommand{\ExampleF}{

\begin{figure}
\begin{center}
\begin{tabular}{c}
\begin{scriptsize}
\xymatrix{ &
*+[F]{\begin{tabular}{c}
    (1) : ``or''-node, $\rBoxSt$\\[0.5ex]
    $\lDmd{\sigma^*}p, [\sigma^*]q, \lnot p \lor \lnot q$
      \end{tabular}}
\ar@{->}[d]
\\
&
*+[F]{\begin{tabular}{c}
    (2) : ``or''-node, $\rDmdSt$\\[0.5ex]
    $\lDmd{\sigma^*}p, q, [\sigma][\sigma^*]q, \lnot p \lor \lnot q$
      \end{tabular}}
\ar@{->}[d] \ar@{->}[dr] &
*+[F]{\begin{tabular}{c}
    (8) : ``and''-node, $\rTrans$\\[0.5ex]
    $\lDmd{\sigma}\lDmd{\sigma^*}p, q, [\sigma][\sigma^*]q, \lnot p$
      \end{tabular}}
\ar_{\lDmd{\sigma}\lDmd{\sigma^*}p}@{->}[ul]
\\
&
*+[F]{\begin{tabular}{c}
    (3) : ``or''-node, $\rOr$\\[0.5ex]
    $p, q, [\sigma][\sigma^*]q, \lnot p \lor \lnot q$
      \end{tabular}}
\ar@{->}[dl] \ar@{->}[d] &
*+[F]{\begin{tabular}{c}
    (4) : ``or''-node, $\rOr$\\[0.5ex]
    $\lDmd{\sigma}\lDmd{\sigma^*}p, q, [\sigma][\sigma^*]q, \lnot p \lor \lnot q$
      \end{tabular}}
\ar@{->}[u] \ar@{->}[d]
\\
*+[F]{\begin{tabular}{c}
    (5) : ``or''-node, $\rBot$\\[0.5ex]
    $p, q, [\sigma][\sigma^*]q, \lnot p$
      \end{tabular}}
\ar@{->}[dr] &
*+[F]{\begin{tabular}{c}
    (6) : ``or''-node, $\rBot$\\[0.5ex]
    $p, q, [\sigma][\sigma^*]q, \lnot q$
      \end{tabular}}
\ar@{->}[d] &
*+[F]{\begin{tabular}{c}
    (9) : ``or''-node, $\rBot$\\[0.5ex]
    $\lDmd{\sigma}\lDmd{\sigma^*}p, q, [\sigma][\sigma^*]q, \lnot q$
      \end{tabular}}
\ar@{->}[dl]
\\
&
*+[F]{\begin{tabular}{c}
    (7)\\[0.5ex]
    $\bot$
      \end{tabular}}
}
\end{scriptsize}
\end{tabular}
\end{center}
\caption{An ``and-or'' graph for $(\{\lDmd{\sigma^*}p, [\sigma^*]q\}, \{\lnot p \lor \lnot q\})$. In the 2nd line of each node we display the formulas of the label of the node. We do not display the sets $\rfs(\_)$ of the nodes.} \label{fig: example1}
\end{figure}
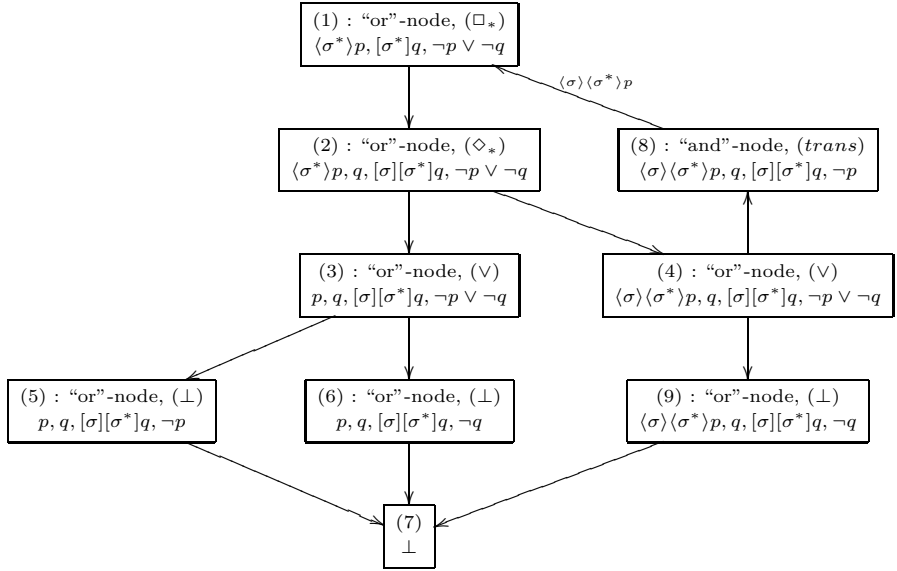

\begin{example}
In Figure~\ref{fig: example1} we give an ``and-or'' graph for $(\{\lDmd{\sigma^*}p, [\sigma^*]q\}$, $\{\lnot p \lor \lnot q\})$. This graph does not have any consistent marking: the only marking that satisfies the local consistency property consists of the nodes (1), (2), (4), (8) and does not satisfy the global consistency property, because the formula $\lDmd{\sigma^*}p$ of the label of (1) does not have any $\Dmd$-realization in this marking. By Theorem~\ref{theorem: s-c-PDL}, the set $\{\lDmd{\sigma^*}p, [\sigma^*]q\}$ is unsatisfiable w.r.t.\ the global assumption $\lnot p \lor \lnot q$.\koniec
\end{example}
} 

\LongVersion{\ExampleF}
\ShortVersion{See Appendix~\ref{appendix: examples} for an example of ``and-or'' graph.}


\VSpace{-0.5em}
\section{A Tableau Calculus for Dealing with ABoxes}\VSpace{-0.5em}
\label{section: ABoxes}

Define tableau calculus \cPDLA\ w.r.t.\ a~TBox $\Gamma$ to be the extension of \cPDL\ with the following additional rules:
\begin{itemize}
\item a~rule $(\rho')$ obtained from each rule $(\rho) \in \{\rAnd$, $\rOr$, $\rBoxSc$, $\rDmdSc$, $\rBoxCp$, $\rDmdCp$, $\rBoxQm$, $\rDmdQm$, $\rBoxSt$, $\rDmdSt\}$ by labeling the principal formula and the formulas obtained from it by prefix ``$a\!:\ $'' and adding the modified principal formula to each of the possible conclusions; for example:
\[ \rOrP\; \fracc{Y,\; a:\varphi \lor \psi}{Y,\; a:\varphi \lor \psi,\; a:\varphi \mid Y,\; a:\varphi \lor \psi,\; a: \psi} \]

\item and
\[
\rBotzP\; \fracc{Y,\; a:\bot}{\bot} \hspace{8ex} \rBotP\; \fracc{Y,\; a:p,\; a:\lnot p}{\bot}
\]
\[
\rBoxP\; \fracc{Y,\; a:[\sigma]\varphi,\; \sigma(a,b)}{Y,\; a:[\sigma]\varphi,\; \sigma(a,b),\; b:\varphi}
\]
\[
\rTransP\; \fracc{Y}{\& \{\; (\{\varphi\} \cup \{\psi \textrm{ s.t. } (a:[\sigma]\psi) \in Y\} \cup \Gamma) \textrm{
s.t. } (a:\lDmd{\sigma}\varphi) \in Y \;\}}
\]
\end{itemize}

The additional rules of \cPDLA\ work on sets of ABox assertions, except that the conclusions of $\rTransP$ are sets of
traditional formulas. That is, in those rules, $Y$ denotes a~set of ABox assertions. The rule $\rTransP$ is an
``and''-rule and a~transitional rule. The other additional rules of \cPDLA\ are ``or''-rules and static rules. 

Note that, for any additional static rule of \cPDLA\ except $\rBotzP$ and $\rBotP$, the premise is a subset of each of the possible conclusions. Such rules are said to be {\em monotonic}.

We assume that any one of the rules $\rAndP$, $\rOrP$, $\rBoxScP$, $\rDmdScP$, $\rBoxCpP$, $\rDmdCpP$, $\rBoxQmP$, $\rDmdQmP$, $\rBoxStP$, $\rDmdStP$ is applicable to a node $v$ only when the principal formula does not belong to $\rfs(v)$. 
Applying any one of these rules to a node $v$, for any successor node $w$ of $v$, let $\rfs(w)$ be the set that extends $\rfs(v)$ with the principal formula of the applied rule. 
We assume that the rule $\rBoxP$ is applicable only when its conclusion is a proper superset of its premise. Applying this rule to a node $v$, let $\rfs(w) = \rfs(v)$ for the successor $w$ of $v$.
Applying $\rBotzP$, $\rBotP$, or $\rTransP$ a node $v$, for any successor node $w$ of $v$, let $\rfs(w) = \emptyset$.

Similarly as for \cPDL, we assume the following preference for the rules of \cPDLA: the rules $\rBotz$, $\rBot$,
$\rBotzP$, $\rBotP$ have the highest priority; unary static rules have a~higher priority that non-unary static rules;
all the static rules have a~higher priority than the transitional rules.

Consider the problem of checking whether a~given ABox $\mA$ is satisfiable w.r.t.\ a~given TBox $\Gamma$. We construct an {\em ``and-or'' graph for} $(\mA,\Gamma)$ as follows. The graph will contain nodes of two kinds: {\em complex nodes} and {\em simple nodes}. The sets $\mL(v)$ and $\rfs(v)$ of a complex node $v$ consist of ABox assertions, while such sets of a simple node $v$ consist of traditional formulas. The graph will never contain edges from a~simple node to a~complex node. The root of the graph is a~complex node $\nu$ with $\mL(\nu) = \mA \cup \{(a\!:\!\varphi) \mid \varphi \in \Gamma$ and $a$ is a~state variable occurring in $\mA\}$ and $\rfs(\nu) = \emptyset$. Complex nodes are expanded using the additional rules of \cPDLA\ (i.e., the ``prime'' rules), while simple nodes are expanded using the rules of \cPDL. The ``and-or'' graph is expanded in the same way as described in the previous section for checking consistency of a~set $X$ of traditional formulas w.r.t.\ $\Gamma$. 

The notion of {\em marking} remains unchanged. 

Let $G$ be an ``and-or'' graph for $(\mA,\Gamma)$ and let $G'$ be a~marking of $G$. If $v$ is a simple node of $G'$ and $\lDmd{\alpha}\varphi$ is a formula of the label of $v$ then a {\em trace} of $\lDmd{\alpha}\varphi$ in $G'$ starting from $v$ is defined as before. Consider the case when $v$ is a complex ``and''-node and suppose that $a\!:\!\lDmd{\alpha}\varphi \in \mL(v)$. A {\em static trace} of $a\!:\!\lDmd{\alpha}\varphi$ at $v$ is a sequence $\varphi_0,\ldots,\varphi_k$ such that:
\begin{itemize}
\item $\varphi_0 = \lDmd{\alpha}\varphi$;
\item for every $1 \leq i \leq k$, $(a\!:\!\varphi_i) \in \mL(v)$;
\item for every $1 \leq i \leq k$, 
  \begin{itemize}
  \item if $\varphi_{i-1} = \lDmd{\beta;\gamma}\psi$ then $\varphi_i = \lDmd{\beta}\lDmd{\gamma}\psi$, 
  \item if $\varphi_{i-1} = \lDmd{\beta\cup\gamma}\psi$ then $\varphi_i$ is either $\lDmd{\beta}\psi$ or $\lDmd{\gamma}\psi$, 
  \item if $\varphi_{i-1} = \lDmd{\psi?}\xi$ then $\varphi_i = \xi$,
  \item if $\varphi_{i-1} = \lDmd{\beta^*}\psi$ then $\varphi_i$ is either $\psi$ or $\lDmd{\beta}\lDmd{\beta^*}\psi$.
  \end{itemize}
\end{itemize}\VSpace{-1em}
A static trace $\varphi_0, \ldots, \varphi_k$ of $a\!:\!\lDmd{\alpha}\varphi$ at $v$ is called a {\em static realization} for $a\!:\!\lDmd{\alpha}\varphi$ at $v$ if either $\varphi_k = \varphi$ or $\varphi_k$ is of the form $\lDmd{\sigma}\varphi'_k$ for some $\sigma \in \mindices$. 

A marking $G'$ of an ``and-or'' graph $G$ for $(\mA,\Gamma)$ is {\em consistent} if:
\begin{description}
\item[local consistency:] \VSpace{-0.8em}$G'$ does not contain any node with label $\{\bot\}$;
\item[global consistency:] \ 
  \begin{itemize}
  \item for every complex ``and''-node $v$ of $G'$, every formula of the form $a\!:\!\lDmd{\alpha}\varphi$ of the label of $v$ has a static realization (at~$v$),
  \item for every simple node $v$ of $G'$, every formula of the form $\lDmd{\alpha}\varphi$ of the label of $v$ has a~$\Dmd$-realization (starting at $v$) in $G'$.
  \end{itemize}
\end{description}

\VSpace{-0.8em}
\begin{theorem}[Soundness and Completeness of \cPDLA]
\label{theorem: sound-compl cPDLA} Let $\mA$ be an ABox, $\Gamma$ a~TBox, and $G$ an ``and-or'' graph for $(\mA,\Gamma)$.
Then $\mA$ is satisfiable w.r.t.\ $\Gamma$ iff $G$ has a~consistent marking.\koniec
\end{theorem}\VSpace{-0.5em}

\LongVersion{The ``only if'' direction means soundness of \cPDLA, while the ``if'' direction means completeness of \cPDLA. See Appendix~\ref{section: proof PDLA} for the proof of this theorem.} 

\ShortVersion{See the full version \cite{pdl-tab-long} for the proof of this theorem.}

\newcommand{\ExampleS}{
\begin{figure}[t]
\begin{center}
\begin{tabular}{c}
\begin{scriptsize}
\xymatrix{ &
*+[F]{\begin{tabular}{c}
    (1) : ``or''-node, $\rBoxP$\\[0.5ex]
    $a: [\sigma]\lDmd{\sigma^*}p,\, \sigma(a,b),$\\
    $a:\lnot p,\, b:\lnot p$
      \end{tabular}}
\ar@{->}[d]
\\
*+[F]{\begin{tabular}{c}
    (3) : ``or''-node, $\rBotP$\\[0.5ex]
    $a: [\sigma]\lDmd{\sigma^*}p,\, \sigma(a,b),$\\
    $a:\lnot p,\, b:\lnot p,\, b:\lDmd{\sigma^*}p,$\\
    $b:p$
      \end{tabular}}
\ar@{->}[d] &
*+[F]{\begin{tabular}{c}
    (2) : ``or''-node, $\rDmdStP$\\[0.5ex]
    $a: [\sigma]\lDmd{\sigma^*}p,\, \sigma(a,b),$\\
    $a:\lnot p,\, b:\lnot p,\, b:\lDmd{\sigma^*}p,$
      \end{tabular}}
\ar@{->}[l] \ar@{->}[d]
\\
*+[F]{\begin{tabular}{c}
    (4)\\[0.5ex]
    $\bot$
      \end{tabular}}
&
*+[F]{\begin{tabular}{c}
    (5) : ``and''-node, $\rTransP$\\[0.5ex]
    $a: [\sigma]\lDmd{\sigma^*}p,\, \sigma(a,b),$\\
    $a:\lnot p,\, b:\lnot p,\, b:\lDmd{\sigma^*}p,$\\
    $b:\lDmd{\sigma}\lDmd{\sigma^*}p$
      \end{tabular}}
\ar^{b:\lDmd{\sigma}\lDmd{\sigma^*}p}@{->}[d]
\\
*+[F]{\begin{tabular}{c}
    (7) : ``or''-node, $\rBot$\\[0.5ex]
    $p, \lnot p$
      \end{tabular}}
\ar@{->}[u] &
*+[F]{\begin{tabular}{c}
    (6) : ``or''-node, $\rDmdSt$\\[0.5ex]
    $\lDmd{\sigma^*}p, \lnot p$
      \end{tabular}}
\ar@{->}[l] \ar@{->}[d]
\\
&
*+[F]{\begin{tabular}{c}
    (8) : ``and''-node, $\rTrans$\\[0.5ex]
    $\lDmd{\sigma}\lDmd{\sigma^*}p, \lnot p$
      \end{tabular}}
\ar_{\lDmd{\sigma}\lDmd{\sigma^*}p}@/^{-1.5pc}/@{->}[u] }
\end{scriptsize}
\end{tabular}
\end{center}
\caption{An ``and-or'' graph for $(\{a: [\sigma]\lDmd{\sigma^*}p,\, \sigma(a,b)\}, \{\lnot p\})$. The formulas in each node $v$ form the set $\mL(v)$. We do not display formulas of the sets $\rfs(v)$.} \label{fig: example2}
\end{figure}
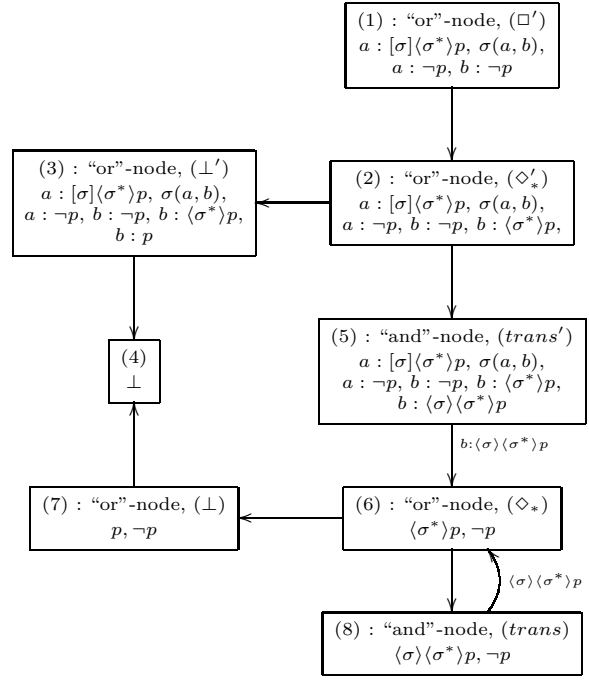

\begin{example}
In Figure~\ref{fig: example2} we present an ``and-or'' graph for $(\{a: [\sigma]\lDmd{\sigma^*}p$,\ $\sigma(a,b)\}$,
$\{\lnot p\})$. This graph does not have any consistent marking. By Theorem~\ref{theorem: sound-compl cPDLA}, the ABox
$\{a: [\sigma]\lDmd{\sigma^*}p,\, \sigma(a,b)\}$ is unsatisfiable w.r.t.\ the TBox $\{\lnot p\}$.\koniec
\end{example}
} 

\LongVersion{\ExampleS}
\ShortVersion{See Appendix~\ref{appendix: examples} for an example of ``and-or'' graph involved with an ABox.}


\VSpace{-0.5em}
\section{Decision Procedures for PDL}\VSpace{-0.8em}
\label{section: dp-cr}

In this section, we present simple algorithms for checking satisfiability of a~given set $X$ of traditional formulas w.r.t.\ a~given set $\Gamma$ of global assumptions and for checking satisfiability of given ABox $\mA$ w.r.t.\ a~given TBox $\Gamma$. 
\LongVersion{Optimizations for the algorithms will be discussed in the next section.}
\ShortVersion{(Optimizations for the algorithms are discussed in~\cite{pdl-tab-long}.)}
We also prove the mentioned data complexity result for PDL. 

Define the {\em length} of a~formula $\varphi$ to be the number of symbols occurring in~$\varphi$, and the {\em size}
of a~finite set of formulas to be the length of the conjunction of its formulas.

\VSpace{-0.8em}
\subsection{Checking Satisfiability of $X$ w.r.t.\ $\Gamma$}\VSpace{-0.5em}

Let $X$ and $\Gamma$ be finite sets of traditional formulas in NNF, $G$ be an ``and-or'' graph for $(X,\Gamma)$, and $G'$ be a~marking of $G$. The {\em graph $G_t$ of traces of $G'$ in $G$} is defined as follows:\VSpace{-0.5em}
\begin{itemize}
\item Nodes of $G_t$ are pairs $(v,\varphi)$, where $v$ is a~node of $G'$ and $\varphi$ is a~formula of the label of $v$.
\item A~pair $((v,\varphi),(w,\psi))$ is an edge of $G_t$ if $v$ is a~node of $G'$, $\varphi$ is of the form $\lDmd{\alpha}\xi$, and the sequence $(v,\varphi)$, $(w,\psi)$ is a~trace of $\varphi$ in $G'$.
\end{itemize}\VSpace{-0.5em}
A node $(v,\varphi)$ of $G_t$ is an {\em end node} if $\varphi$ is not of the form $\lDmd{\alpha}\xi$. A node of $G_t$
is {\em productive} if there is a~path connecting it to an end node.

Consider now Algorithm~\ref{alg1} (see Figure~\ref{alg:1}) for checking satisfiability of $X$ w.r.t.\ $\Gamma$.
The algorithm starts by constructing an ``and-or'' graph $G$ with root $v_0$ for $(X,\Gamma)$. After that it collects the nodes of $G$ whose labels are unsatisfiable w.r.t.\ $\Gamma$. Such nodes are said to be \Unsat\ and kept in the set
$UnsatNodes$. Initially, if $G$ contains a~node with label $\{\bot\}$ then the node is \Unsat. When a~node or a~number of
nodes become \Unsat, the algorithm propagates the status \Unsat\ backwards through the ``and-or'' graph using the procedure
$updateUnsatNodes$ (see Figure~\ref{alg:1}). This procedure has property that, after calling it, if the root $v_0$ of $G$ does not belong to
$UnsatNodes$ then the maximal subgraph of $G$ without nodes from $UnsatNodes$, denoted by $G'$, is a~marking of $G$.
After each calling of $updateUnsatNodes$, the algorithm finds the nodes of $G'$ that make the marking not satisfying
the global consistency property. Such a~task is done by creating the graph $G_t$ of traces of $G'$ in $G$ and finding nodes $v$ of $G'$ such that the label of $v$ contains a~formula of the form $\lDmd{\alpha}\varphi$ but $(v,\lDmd{\alpha}\varphi)$ is not a~productive node of~$G_t$. If the set $V$ of such nodes is empty then $G'$ is a~consistent marking (provided that $v_0 \notin UnsatNodes$) and the algorithm stops with a~positive answer. Otherwise,
$V$ is used to update $UnsatNodes$ by calling $updateUnsatNodes(G,UnsatNodes,V)$. After that call, if $v_0 \in
UnsatNodes$ then the algorithm stops with a~negative answer, else the algorithm repeats the loop of collecting \Unsat\
nodes. Note that, we can construct $G_t$ only the first time and update it appropriately each time when $UnsatNodes$ is
changed.

\begin{figure}[t!]
\ramka{
\begin{algorithm} \label{alg1} \ \\
Input: finite sets $X$ and $\Gamma$ of traditional formulas in NNF.\\
Output: {\em true} if $X$ is satisfiable w.r.t.\ $\Gamma$, and {\em false} otherwise.
\newcommand{\mIndent}{\mbox{\hspace{1em}}}
\begin{enumerate}
\item construct an ``and-or'' graph $G$ with root $v_0$ for $(X,\Gamma)$;
\item $UnsatNodes := \emptyset$;
\item if $G$ contains a~node $v$ with label $\{\bot\}$ then\\
\mIndent $updateUnsatNodes(G,UnsatNodes,\{v\})$;
\item if $v_0 \in UnsatNodes$ then return {\em false};
\item let $G'$ be the maximal subgraph of $G$ without nodes from $UnsatNodes$;\\
(we have that $G'$ is a~marking of $G$)
\item construct the graph $G_t$ of traces of $G'$ in $G$;
\item while $v_0 \notin UnsatNodes$ do:
  \begin{enumerate}
  \item \label{alg1-step-comp-nodes-not-satisfying-global-consistency} let $V$ be the set of all nodes $v$ of $G'$ such that $G_t$ contains a non-productive node of the form $(v,\lDmd{\alpha}\varphi)$;
  \item \label{alg1-step-exit-true} if $V = \emptyset$ then return {\em true};
  \item $updateUnsatNodes(G,UnsatNodes,V)$;
  \item if $v_0 \in UnsatNodes$ then return {\em false};
  \item let $G'$ be the maximal subgraph of $G$ without nodes from $UnsatNodes$;\\ (we have that $G'$ is a~marking of $G$)
  \item update $G_t$ to the graph of traces of $G'$ in $G$;
  \end{enumerate}\VSpace{-0.5em}
\end{enumerate}
\end{algorithm}
\LongVersion{\smallskip}
\flushleft {\bf Procedure } $updateUnsatNodes(G,UnsatNodes,V)$\\
Input: an ``and-or'' graph $G$ and sets $UnsatNodes$, $V$ of nodes of~$G$,\\
\mbox{\hspace{2.7em}} where $V$ contains new \Unsat\ nodes.\\
Output: a~new set $UnsatNodes$.
\begin{enumerate}
\item $UnsatNodes := UnsatNodes \cup V$;
\item while $V$ is not empty do:
  \begin{enumerate}
  \item take out a~node $v$ from $V$;
  \item for every father node $u$ of $v$, if $u \notin UnsatNodes$ and either $u$ is an ``and''-node or $u$ is an ``or''-node and all the successor nodes of $u$ belong to $UnsatNodes$ then add $u$ to both $UnsatNodes$ and $V$;
  \end{enumerate}
\end{enumerate}
}\VSpace{-1em}
\caption{Algorithm for checking satisfiability of $X$ w.r.t.\ $\Gamma$.\label{alg:1}}\VSpace{-1.5em}
\end{figure}

\LongVersion{
\begin{lemma} \label{lemma: size constr and-or gr}
Let $X$ and $\Gamma$ be finite sets of traditional formulas in NNF, $G$ be an ``and-or'' graph for $(X,\Gamma)$, and $n$ be the size of $X \cup \Gamma$. Then $G$ has $2^{O(n)}$ nodes, and for each node $v$ of $G$, the sets $\mL(v)$ and $\rfs(v)$ contain at most $O(n)$ formulas and are of size $O(n^2)$.
\end{lemma}

\begin{proof} 
The sets $\mL(v)$ and $\rfs(v)$ of each node $v$ of $G$ are subsets of the Fischer-Ladner closure $FL(X \cup \Gamma)$. This closure contains at most $O(n)$ formulas \cite[Lemma 6.3]{HKT00}.\footnote{In~\cite{HKT00}, only $\bot$, $\to$, $[\alpha]$ are considered as primitive, while $\lnot$, $\land$, $\lor$, $\lDmd{\alpha}$ are treated as derived operators. However, the lemma still holds for our language.}
Hence $\mL(v)$ and $\rfs(v)$ contain at most $O(n)$ formulas and are of size $O(n^2)$. Since the nodes of $G$ have unique contents, $G$ has $2^{O(n)}$ nodes. \koniec
\end{proof}

\begin{lemma} \label{lemma: alg1 compl}
Algorithm~\ref{alg1} runs in exponential time in the size of $X \cup \Gamma$. 
\end{lemma}

\begin{proof}
By Lemma~\ref{lemma: size constr and-or gr}, the graph $G$ can be constructed in $2^{O(n)}$ steps and has $2^{O(n)}$ nodes. As the label of each node of $G$ contains at most $O(n)$ formulas, each time when $UnsatNodes$ is extended $G_t$ can be constructed or updated in $2^{O(n)}$ steps. Computing the set $V$ can be done in polynomial time in the size of $G_t$, and hence also in $2^{O(n)}$ steps. An execution of $updateUnsatNodes$ is done in polynomial time in the size of $G$, and hence also in $2^{O(n)}$ steps. As the set $UnsatNodes$ is extended at most $2^{O(n)}$ times, the total time for executing Algorithm~\ref{alg1} is of rank $2^{O(n)}$.\koniec
\end{proof}
} 

\begin{theorem} \label{theorem: PDL}
Let $X$ and $\Gamma$ denote finite sets of traditional formulas in NNF. Algorithm~\ref{alg1} is an \EXPTIME\ decision procedure for checking whether $X$ is satisfiable w.r.t.\ the set $\Gamma$ of global assumptions.
\end{theorem}
\VSpace{-1em}
\begin{proof}
It is easy to show that the algorithm has the invariant that a~consistent marking of $G$ cannot contain any node of
$UnsatNodes$. The algorithm returns {\em false} only when the root $v_0$ belongs to $UnsatNodes$, that is, only when
$G$ does not have any consistent marking. At Step~\ref{alg1-step-exit-true}, $G'$ is a~marking of $G$ that satisfies
the local consistency property. If at that step $V = \emptyset$ then it satisfies also the global consistency property
and is thus a~consistent marking of $G$. That is, the algorithm returns {\em true} only when $G$ has a~consistent
marking. Therefore, by Theorem~\ref{theorem: s-c-PDL}, Algorithm~\ref{alg1} is a~decision procedure for the considered
problem. 
\LongVersion{The complexity was established by Lemma~\ref{lemma: alg1 compl}.}
\ShortVersion{See the full version~\cite{pdl-tab-long} for an analysis of the complexity of the algorithm.}
\koniec
\end{proof}

\VSpace{-1em}
\subsection{Checking Satisfiability of an ABox w.r.t.\ a~TBox}\VSpace{-0.5em}

Let $\mA$ be an ABox, $\Gamma$ be a~TBox, $G$ be an ``and-or'' graph for $(\mA,\Gamma)$, and $G'$ be a~marking of $G$. The {\em graph $G_t$ of traces of $G'$ in $G$} is the largest graph such that:
\begin{itemize}
\item A node of $G_t$ is
  \begin{itemize}
  \item either a pair $(v,\varphi)$, where $v$ is a simple node of $G'$ and $\varphi \in \mL(v)$, 
  \item or a pair $(v,a\!:\!\varphi)$, where $v$ is a complex ``and''-node of $G'$ and $(a\!:\!\varphi) \in \mL(v)$.
  \end{itemize}

\item An edge of $G_t$ is
  \begin{itemize}
  \item either a pair $((v,\varphi),(w,\psi))$ such that $v$ is a simple node of $G'$, $\varphi$ is of the form $\lDmd{\alpha}\xi$, and the sequence $(v,\varphi)$, $(w,\psi)$ is a~trace of $\varphi$ in~$G'$,
  \item or a pair $((v,a\!:\!\varphi),(v,a\!:\!\psi))$ such that $v$ is a complex ``and''-node of $G'$, $\varphi$ is of the form $\lDmd{\alpha}\xi$, and the sequence $\varphi$, $\psi$ is a static trace of $a\!:\!\varphi$ at~$v$,
  \item or a pair $((v,a\!:\!\lDmd{\sigma}\varphi),(w,\varphi))$ such that $v$ is a complex ``and''-node of $G'$ and $(v,w)$ is an edge of $G'$ with $a\!:\!\lDmd{\sigma}\varphi$ as the label.
  \end{itemize}
\end{itemize}

A node of $G_t$ is an {\em end node} if it is of the form $(v,\varphi)$ or $(v,a\!:\!\varphi)$ such that $\varphi$ is not of the form $\lDmd{\alpha}\xi$. A node of $G_t$ is {\em productive} if there is a~path connecting it to an end node.

By {\em Algorithm~\ref{alg1}$'$} we refer to the algorithm obtained from Algorithm~\ref{alg1} by changing $X$ to $\mA$ and modifying Step~\ref{alg1-step-comp-nodes-not-satisfying-global-consistency} to ``let $V$ be the set of all nodes $v$ of $G'$ such that $G_t$ contains a non-productive node of the form $(v,\lDmd{\alpha}\varphi)$ or $(v,a\!:\!\lDmd{\alpha}\varphi)$''.
Algorithm~\ref{alg1}$'$ receives an ABox $\mA$ and a~TBox $\Gamma$ as input and checks whether $\mA$ is satisfiable w.r.t.~$\Gamma$.

\LongVersion{Here is a~counterpart of Lemma~\ref{lemma: size constr and-or gr}:}

\VSpace{-0.5em}
\begin{lemma} \label{lemma: size constr and-or gr ABox}
Let $\mA$ be an ABox, $\Gamma$ be a~TBox, $G$ be an ``and-or'' graph for $(\mA,\Gamma)$, and $n$ be the size of $\mA \cup \Gamma$. Then $G$ has $2^{O(n^2)}$ nodes. If $v$ is a simple node of $G$ then $\mL(v)$ and $\rfs(v)$ contain at most $O(n)$ formulas and are of size $O(n^2)$. If $v$ is a complex node of $G$ then $\mL(v)$ and $\rfs(v)$ contain at most $O(n^2)$ formulas and are of size $O(n^3)$.
\end{lemma}\VSpace{-0.5em}

\LongVersion{
\begin{proof}
Let $S$ be the set of all state variables occurring in $\mA$ and let $X = \Gamma \cup \{ \varphi \mid (a:\varphi) \in \mA \textrm{ for some } a~\in S\}$. The sets $\mL(v)$ and $\rfs(v)$ of each simple node $v$ of $G$ are subsets of the Fischer-Ladner closure $FL(X)$. Since this closure contains at most $O(n)$ formulas \cite[Lemma 6.3]{HKT00}, the sets $\mL(v)$ and $\rfs(v)$ of each simple node $v$ of $G$ contain at most $O(n)$ formulas and are of size $O(n^2)$. Since the simple nodes of $G$ have unique contents, $G$ has $2^{O(n)}$ simple nodes. 
For each complex node $v$ of $G$ and for each $a \in S$, the set $\{\varphi \mid (a:\varphi) \in \mL(v) \cup \rfs(v)\}$ is also a~subset of $FL(X)$. Hence the sets $\mL(v)$ and $\rfs(v)$ of each complex node $v$ of $G$ contain at most $O(n^2)$ formulas and are of size $O(n^3)$. Due to the restrictions on applicability of the static ``prime'' rules of \cPDLA, each path of complex nodes in $G$ has length of rank $O(n^2)$. Hence $G$ contains $2^{O(n^2)}$ complex nodes. \koniec
\end{proof}
} 

\LongVersion{Using the proofs of Lemma~\ref{lemma: alg1 compl} and Theorem~\ref{theorem: PDL} with appropriate changes we obtain the following theorem.}
\ShortVersion{See~\cite{pdl-tab-long} for the proof of this lemma. Using this lemma and the proof of Theorem~\ref{theorem: PDL} with appropriate changes (see also~\cite[Lemma~6.2]{pdl-tab-long}) we obtain the following theorem.} 

\VSpace{-0.2em}
\begin{theorem}
Algorithm~\ref{alg1}$'$ is an \EXPTIME\ decision procedure for checking whether a~given ABox $\mA$ is satisfiable
w.r.t.\ a~given TBox $\Gamma$. \koniec
\end{theorem}
\VSpace{-0.2em}

Algorithm~\ref{alg1}$'$ uses global caching for both complex nodes and simple nodes. What happens if we use global
caching only for simple nodes and backtracking on branchings at complex ``or''-nodes? Is the complexity still \EXPTIME?
The rest of this subsection deals with these questions.

\VSpace{-0.4em}
\begin{lemma} \label{lemma: backtrack ABox}
Let $\mA$ be an ABox, $\Gamma$ be a~TBox, and $G$ be an ``and-or'' graph for $(\mA,\Gamma)$. Then $G$ has a~consistent marking iff there exists a~complex ``and''-node $v$ of $G$ such that the subgraph generated by $v$ of $G$ (which uses $v$ as
the root) has a~consistent marking.
\end{lemma}
\VSpace{-1.0em}
\begin{proof}
Just notice that the root of $G$ is a~complex node and every father node of a~complex node must be a~complex
``or''-node. \koniec
\end{proof}

By {\em Algorithm~\ref{alg1}$''$} we refer to the algorithm that checks whether a~given ABox $\mA$ is satisfiable
w.r.t. a~given TBox $\Gamma$ as follows. The algorithm ``simulates'' the tasks of constructing an ``and-or'' graph for
$(\mA,\Gamma)$ and checking whether the graph has a~consistent marking but does it as follows:
\begin{enumerate}
\item nondeterministically expand a~path from the root until reaching a~complex ``and''-node~$v$;
\item construct the full subgraph rooted at $v$;
\item check whether the subgraph has a~consistent marking (as done in the steps 2--7 of Algorithm~\ref{alg1}), and return {\em true} if it does;
\item if none of the possible executions returns {\em true} then return {\em false}.
\end{enumerate}

In practice, the first step of the above algorithm is executed by backtracking on the branchings of the applications of
``or''-rules. The algorithm does not keep all complex nodes but only the ones on the current path of complex nodes. On
the other hand, simple nodes can be globally cached. That is, simple nodes can be left through backtracking for use in
the next possible executions.

\begin{theorem}
Using backtracking to deal with nondeterminism, Algorithm~\ref{alg1}$''$ is an \EXPTIME\ decision procedure for checking whether a~given ABox $\mA$ is satisfiable w.r.t.\ a~given TBox $\Gamma$.
\end{theorem}
\VSpace{-1.5em}
\begin{proof}
\ShortVersion{(Sketch)\ }
By Theorem~\ref{theorem: sound-compl cPDLA} and Lemma~\ref{lemma: backtrack ABox}, Algorithm~\ref{alg1}$''$ is a~decision procedure for the considered problem. It remains to show that the algorithm runs in exponential time. Let $n$ be the size of $\mA \cup \Gamma$. As stated in the proof of \LongVersion{Lemma~\ref{lemma: size constr and-or gr ABox}}\ShortVersion{~\cite[Lemma 6.4]{pdl-tab-long}}, each path of complex nodes constructed by Step~1 of Algorithm~\ref{alg1}$''$ has length of rank $O(n^2)$. Analogously to the proofs of 
\LongVersion{Lemmas~\ref{lemma: alg1 compl} and \ref{lemma: size constr and-or gr ABox},}
\ShortVersion{\cite[Lemmas~6.2 and~6.4]{pdl-tab-long},}
it can be shown that Steps 2 and 3 of Algorithm~\ref{alg1}$''$ are executed in $2^{O(n)}$ steps. Hence the complexity of Algorithm~\ref{alg1}$''$ is of rank $2^{O(n^2)} \times 2^{O(n)}$, which is $2^{O(n^2)}$. \koniec
\end{proof}


\VSpace{-1.5em}
\subsection{On the Instance Checking Problem}
\label{section: icp}
\VSpace{-0.5em}

Observe that $(\mA,\Gamma) \models \varphi(a)$ iff the ABox $\mA \cup \{a:\overline{\varphi}\}$ is unsatisfiable
w.r.t.~$\Gamma$. So, the instance checking problem is reduced to the problem of checking unsatisfiability of an ABox
w.r.t.\ a~TBox. What we are interested in is the {\em data complexity} of the instance checking problem, which is
measured in the size of $\mA$ when assuming that $\mA$ is extensionally reduced and $\Gamma$, $\varphi$, $a$ are fixed.
Here, $\Gamma$, $\varphi$ and $a$ form a~fixed query, while $\mA$ varies as input data.

\begin{theorem}
The data complexity of the instance checking problem in PDL is coNP-complete.
\end{theorem}
\VSpace{-1em}
\begin{proof}
Let $\mA$ be an extensionally reduced ABox, $\Gamma$ be a~TBox, $\varphi$ be a~(traditional) formula in NNF, and $a$ be a~state variable. Consider the problem of checking whether $(\mA,\Gamma) \models \varphi(a)$.

Let $p$ be a~fresh proposition (not occurring in $\mA$, $\Gamma$, $\varphi$) and let $\Gamma'  =  \Gamma \cup \{\lnot p \lor \varphi$, $p \lor \overline{\varphi}\}$ and $\mA'  =  \mA \cup \{a:\lnot p\}$.

Observe that $\Gamma'$ extends $\Gamma$ with the formulas stating that $p$ is equivalent to~$\varphi$, and that
$(\mA,\Gamma) \models \varphi(a)$ iff the ABox $\mA'$ is unsatisfiable w.r.t.\ the TBox~$\Gamma'$.

Let $n$ be the size of $\mA$. The size of $\mA' \cup \Gamma'$ is thus of rank $O(n)$.

Consider an execution of Algorithm~\ref{alg1}$''$ for the pair $\mA'$ and $\Gamma'$. As stated in the proof of \LongVersion{Lemma~\ref{lemma: size constr and-or gr ABox}}\ShortVersion{~\cite[Lemma 6.4]{pdl-tab-long}}, each path of complex nodes constructed by Step~1 of Algorithm~\ref{alg1}$''$ has length of rank $O(n^2)$. The sets $\mL(v)$ and $\rfs(v)$ of each complex node contain at most $O(n^2)$ formulas. Hence a~nondeterministic execution of Step~1 of Algorithm~\ref{alg1}$''$ runs in time $O(n^2) \times O(n^2)$. Since $A'$ is extensionally reduced, The sets $\mL(v)$ and $\rfs(v)$ of each simple node $v$ depends only on $\Gamma'$. Since $\Gamma'$ is fixed, Steps 2 and 3 of Algorithm~\ref{alg1}$''$ are executed in time of rank $O(n^2)$. Hence the execution of Algorithm~\ref{alg1}$''$ for $\mA'$ and $\Gamma'$ runs nondeterministically in polynomial time the size of $\mA$, and therefore the instance checking problem $(\mA,\Gamma) \models \varphi(a)$ is in coNP.

The coNP-hardness follows from the fact that the instance checking problem in the description logic $\ALC$ is coNP-hard (see~\cite{Schaerf94}). \koniec
\end{proof}


\LongVersion{
\section{Optimizations}
\label{section: opts}

In this section we discuss optimizations for the algorithms given in the previous section. For simplicity we consider
only Algorithm~\ref{alg1}, but the optimizations are applicable also to Algorithms~\ref{alg1}$'$ and \ref{alg1}$''$.

Observe that Algorithm~\ref{alg1} first constructs an ``and-or'' graph and then checks whether the graph contains
a~consistent marking. To speed up the performance these two tasks can be done concurrently. For this we update the
structures $UnsatNodes$, $G'$, $G_t$ mentioned in the algorithm ``on-the-fly'' during the construction of $G$. The main changes are as follows:
\begin{itemize}
\item During the construction of the ``and-or'' graph $G$, each node of $G$ has status \Unexpanded, \Expanded, \Unsat\ or \Sat. The initial status of a~new node is \Unexpanded. When a~node is expanded, we change its status to \Expanded. The status of a~node changes to \Unsat\ (resp.\ \Sat) when there is an evidence that the label of the node is unsatisfiable (resp.\ satisfiable) w.r.t.\ $\Gamma$. When a~node becomes \Unsat, we insert it into the set $UnsatNodes$.
\item When a~node of $G$ is expanded or $G'$ is modified, we update $G_t$ appropriately.
\item When a~new node is created, if its label contains $\bot$ or a~clashing pair $\varphi$, $\overline{\varphi}$ then we change the status of the node to \Unsat. This is the implicit application of the rule $\rBotz$ and a~generalized form of the rule $\rBot$. Thus, we can drop the explicit rules $\rBotz$ and $\rBot$. When a~non-empty set $V$ of nodes of $G$ becomes \Unsat, we call $updateUnsatNodes(G,UnsatNodes,V)$ to update the set $UnsatNodes$.
\item When $UnsatNodes$ is modified, we update $G'$ appropriately.
\item Since $G_t$ is not completed during the construction, when computing the set $V$ of nodes of $G'$ that cause $G'$ not satisfying the global consistency property as in Step~\ref{alg1-step-comp-nodes-not-satisfying-global-consistency} of Algorithm~\ref{alg1} we treat a~node $(v,\varphi)$ of $G_t$ also as an {\em end-node} if $v$ has status \Unexpanded\ or \Sat.\footnote{Note that if $v$ has status \Unexpanded\ (resp.\ \Sat) then $(v,\varphi)$ may (resp.\ must) be a~productive node of $G_t$.} We compute such a~set $V$ occasionally, accordingly to some criteria, and when $G_t$ has been completed. The computation is done by propagating ``productiveness'' backward through the graph $G_t$. The nodes of the resulting $V$ become $\Unsat$.
\end{itemize}

During the construction of the ``and-or'' graph $G$, if a~subgraph of $G$ has been fully expanded in the sense that none of its nodes has status \Unexpanded\ or has a descendant node with status \Unexpanded\ then each node of the subgraph can be determined to be \Unsat\ or \Sat\ regardlessly of the rest of $G$. That is, if a~node of the subgraph cannot be determined to be \Unsat\ by the operations described in the above list then we can set its status to \Sat. This technique was proposed in~\cite{Nguyen08CSP-FI}.

A number of optimizations developed by previous researchers (see, e.g., \cite{HorrocksP99,donini-massacci-exptime-tableau-for-alc}) can be applied for our algorithms. Apart from that, a number of special optimization techniques for search space of the form of ``and-or'' graphs has been developed~\cite{GoreNguyenGC,Nguyen08CSP-FI}. These optimizations have been implemented and experimented with by the first author for the tableau prover TGC for checking satisfiability in $\ALC$~\cite{Nguyen08CSP-FI}.\footnote{Only a simple kind of absorption optimization has been implemented for TGC: for the case the TBox is acyclic and consists of only concept definitions of the form $A \doteq C$, ``lazy unfolding'' is used; consequently, TGC runs on the test set DL'98 T98-kb equally well as on the test set DL'98 T98-sat. For the case the TBox is acyclic and contains also concept inclusions of the form $A \sqsubseteq C$, a simple solution can be adopted: treat $A \sqsubseteq C$ as $A \doteq (C \sqcap A')$ for a new atomic concept $A'$. For the case the TBox is cyclic, one can try to divide the TBox into two parts $\mathcal{T}_1 \cup \mathcal{T}_2$, where $\mathcal{T}_1$ is a maximal acyclic sub-TBox ``not depending'' on the concepts defined in $\mathcal{T}_2$, then one can apply the mentioned ``replacing'' and ``lazy unfolding'' techniques for $\mathcal{T}_1$. Of course, more advanced absorption optimizations can also be tried for TGC. (Here, we write about TGC, but note that TGC can be extended for dealing with PDL in a natural way.)} 
The experimental results of TGC show that some of them are very useful. Most of the optimization techniques discussed in \cite{GoreNguyenGC,Nguyen08CSP-FI} can directly be applied for PDL. However, two things need be further worked out for PDL. The first one is how to efficiently compute ``\Unsat-core'' of a~node that becomes \Unsat\ because it violates the global consistency property.\footnote{An \Unsat-core of a~node is a~subset of the label of the node that causes the node \Unsat. The smaller an \Unsat-core, the better its usefulness (for subset-checking).} The second one is what normalized form should be used for formulas in PDL. It is not difficult to give some solutions for these problems, but their usefulness should be estimated by tests.
} 

\VSpace{-1.0em}
\section{Conclusions}
\label{section: conc}
\VSpace{-0.5em}

In this paper we first provided a tableau-based algorithm for checking satisfiability of a~set of formulas in PDL. We then gave an \EXPTIME\ tableau decision procedure for checking consistency of an ABox w.r.t.\ a~TBox in PDL ($\ALC_{reg}$). 

Our latter procedure is the first optimal (\EXPTIME) tableau decision procedure not based on transformation for checking consistency of an ABox w.r.t.\ a TBox in PDL. Recall that, in~\cite{DeGiacomoThesis} the ABox is encoded by nominals, while in~\cite{GiacomoL96} the ABox is encoded by a concept assertion plus terminology axioms. Note that the approach based on transformation is not efficient in practice: in the well-known tutorial ``Description Logics - Basics, Applications, and More'', Horrocks and Sattler wrote {\em``direct algorithm/implementation instead of encodings''} and {\em``even simple domain encoding is disastrous with large numbers of roles''}. 

The result that the data complexity of the instance checking problem in PDL is coNP-complete is first established in our paper. 

Combining global caching for nodes representing objects not occurring as individuals in the ABox with backtracking for nodes representing individuals occurring in the ABox to obtain another \EXPTIME\ decision procedure is first studied by us in this paper.

Despite that our decision procedure for the case without ABoxes is based on Pratt's algorithm for PDL, our formulation of the tableau calculus for the procedure and our proof of its completeness are completely different than the ones of Pratt. Our decision procedure is formulated in a much simpler way. Note that Pratt's algorithm has been considered complicated: Donini and Massacci wrote in their paper~\cite{donini-massacci-exptime-tableau-for-alc} on \EXPTIME\ tableaux for $\ALC$ that they had proposed ``the first simple tableau based decision procedure working in single exponential time'' (here, note that $\ALC$ is a sub-logic of PDL), which in turn is considered by Baader and Sattler~\cite{BaaderSattler01} still as ``rather complicated''. Also note that nobody has implemented Pratt's algorithm (except Pratt himself in the 70s, but his prototype is not available) and it is natural to ask why that algorithm, known since the 70s, remains unimplemented. 

The idea of global caching comes from Pratt's paper on PDL, but it was discussed rather informally. Donini and Massacci in the mentioned paper on $\ALC$ stated that the caching optimization technique {\em``prunes heavily the search space but its unrestricted usage may lead to unsoundness [37]. It is conjectured that `caching' leads to \EXPTIME-bounds but this has not been formally proved so far, nor the correctness of caching has been shown.''}. Gor{\'e} and Nguyen have recently formalized sound global caching \cite{GoreNguyenTab07,GoreNguyen08CSP} for tableaux in a number of modal logics without the $*$ operator. Extending sound global caching for PDL would better be ``formally proved'' as done in our paper. Our extension for PDL considerably differs from \cite{GoreNguyenTab07,GoreNguyen08CSP}~: 
\LongVersion{\begin{itemize}}
\LongVersion{\item}\ShortVersion{i)} Due to the $*$ operator we have to check not only local consistency but also global consistency of the constructed ``and-or'' graph.
\LongVersion{\item}\ShortVersion{ii)} We defined tableaux directly as ``and-or'' graphs with global caching, while in \cite{GoreNguyenTab07,GoreNguyen08CSP} Gor{\'e} and Nguyen used (traditional) tree-like tableaux and formulated global caching separately. Consequently, we do not have to prove soundness of global caching when having soundness and completeness of the calculus, while Gor{\'e} and Nguyen \cite{GoreNguyenTab07,GoreNguyen08CSP} had to prove soundness of global caching separately after having completeness of their calculi. 
\LongVersion{\end{itemize}}

Our method is applicable for other modal logics, e.g.\ CPDL and regular grammar logics with/without converse. As consequences, it can be shown that the data complexity of the instance checking problem in these logics is coNP-complete. 


\bigskip
\noindent
{\bf Acknowledgements }
We would like to thank Rajeev Gor{\'e} and Florian Widmann for pointing out a mistake in the previous version of this paper. 


\VSpace{-0.5em}


\ShortVersion{
\newpage
\appendix

\section{Examples of ``And-Or'' Graphs}
\label{appendix: examples}

\ExampleF

\ExampleS
} 

\LongVersion{
\newpage
\appendix

\section{Soundness and Completeness of \cPDL}
\label{section: proof PDL}

The alphabet $\Sigma(\alpha)$ of a~program $\alpha$ is defined as follows: $\Sigma(\sigma) = \{\sigma\}$,
$\Sigma(\varphi?) = \{\varphi?\}$, $\Sigma(\beta;\gamma) = \Sigma(\beta) \cup \Sigma(\gamma)$, $\Sigma(\beta \cup
\gamma) = \Sigma(\beta) \cup \Sigma(\gamma)$, $\Sigma(\beta^*) = \Sigma(\beta)$. Thus, $\Sigma(\alpha)$ contains not
only atomic programs but also expressions of the form~$\varphi?$.

A program $\alpha$ is a~regular expression over its alphabet $\Sigma(\alpha)$. The regular language $\mL(\alpha)$
generated by $\alpha$ is specified as follows: $\mL(\sigma) = \{\sigma\}$, $\mL(\varphi?) = \{\varphi?\}$, $\mL(\beta
\cup \gamma) = \mL(\beta) \cup \mL(\gamma)$, $\mL(\beta;\gamma) = \mL(\beta).\mL(\gamma)$, and $\mL(\beta^*) =
(\mL(\beta))^*$, where if $L$ and $M$ are sets of words then $L.M  = \{\alpha\beta \mid \alpha \in L, \beta \in M\}$
and $L^* = \bigcup_{n \geq 0} L^n$ with $L^0 = \{\varepsilon\}$ and $L^{n+1} = L.L^n$, where $\varepsilon$ is the empty
word. We treat words of $\mL(\alpha)$ also as programs, e.g.\ $\sigma_1(\varphi?)\sigma_2$ as
$(\sigma_1;\varphi?;\sigma_2)$.

By $\mG(\alpha)$ we denote the context-free grammar over alphabet $\Sigma(\alpha)$, which is specified as follows: the
grammar variables are sub-expressions of $\alpha$ that do not belong to $\Sigma(\alpha)$, the starting symbol is
$\alpha$, and the grammar rules are:
\begin{eqnarray*}
(\beta;\gamma) & \to & \beta\gamma \\
(\beta \cup \gamma) & \to & \beta \mid \gamma \\
(\beta^*) & \to & \varepsilon \mid \beta(\beta^*)
\end{eqnarray*}

Given a~modality $\triangle$ of the form $\lDmd{\alpha_1}\ldots\lDmd{\alpha_k}$ or $[\alpha_1]\ldots[\alpha_k]$ we call $\alpha_1\ldots\alpha_k$ the program sequence corresponding to~$\triangle$. We call $\lDmd{\alpha_1}\ldots\lDmd{\alpha_k}$ (resp.\ $[\alpha_1]\ldots[\alpha_k]$) the existential (resp.\ universal) modality corresponding to $\alpha_1\ldots\alpha_k$. 

\subsection{Soundness}

\begin{lemma}[Soundness] \label{lemma: soundness PDL}
Let $X$ and $\Gamma$ be finite sets of traditional formulas in NNF, and $G$ be an ``and-or'' graph for $(X,\Gamma)$. Suppose that $X$ is
satisfiable w.r.t.\ the set $\Gamma$ of global assumptions. Then $G$ has a~consistent marking.
\end{lemma}

\begin{proof}
We construct a~consistent marking $G'$ of $G$ as follows. At the beginning, $G'$ contains only the root of $G$. Then,
for every node $v$ of $G'$ and for every successor $w$ of $v$ in $G$, if the label of $w$ is satisfiable w.r.t.\
$\Gamma$, then add the node $w$ and the edge $(v,w)$ to $G'$. It is easy to see that $G'$ is a~marking of $G$. Also,
$G'$ clearly satisfies the local consistency property.

We now check the global consistency property of $G'$. Let $v_0$ be a~node of $G'$, $Y$ be the label of $v_0$, and
$\lDmd{\alpha}\varphi$ be a~formula of $Y$. We show that the formula has a~$\Dmd$-realization (starting from $v_0$) in $G'$. As $Y$ is
satisfiable w.r.t.\ $\Gamma$, there exists a~Kripke model $\mM$ that validates $\Gamma$ and satisfies $Y$ at a~state
$u_0$.
Since $\lDmd{\alpha}\varphi$ is satisfied at $u_0$ in $\mM$, there exist a~word $\delta = \sigma_1\ldots
\sigma_{j_1}(\psi_1?)\sigma_{j_1+1}\ldots \sigma_{j_2}(\psi_2?)\ldots \sigma_{j_k} \in \mL(\alpha)$ (with $0 \leq j_1
\leq j_2 \leq \ldots \leq j_k$) and states $u_1,\ldots,u_{j_k}$ of $\mM$ such that: $(u_{j-1},u_j) \in \sigma_j^\mM$
for $1 \leq j \leq j_k$, $u_{j_l} \in \psi_l^\mM$ for $1 \leq l \leq k-1$, and $u_{j_k} \in \varphi^\mM$. Denote this
property by~$(\star)$.

We construct a~$\Dmd$-realization $(v_0,\varphi_0),\ldots,(v_h,\varphi_h)$ in $G'$ for $\lDmd{\alpha}\varphi$ at $v_0$
and a~mapping $f : \{v_0,\ldots,v_h\} \to \{u_0,\ldots,u_{j_k}\}$ such that $f(v_0) = u_0$, $f(v_h) = u_{j_k}$, and for
every $0 \leq i < h$, if $f(v_i) = u_j$ then $f(v_{i+1})$ is either $u_j$ or $u_{j+1}$.
For $1 \leq i \leq h$, let $\triangle_i$ be the sequence of existential modal operators such that $\varphi_i =
\triangle_i\varphi$ and let $S_i$ be the program sequence corresponding to $\triangle_i$.
We maintain the following invariants for $0 \leq i \leq h$~:
\begin{enumerate}
\item[(a)] The sequence $(v_0,\varphi_0),\ldots,(v_i,\varphi_i)$ is a~trace of $\lDmd{\alpha}\varphi$ in $G'$.
\item[(b)] The label of $v_i$ is satisfied at the state $f(v_i)$ of $\mM$.
\item[(c)] If $f(v_i) = u_j$ and $j \notin \{j_1,\ldots,j_{k-1}\}$ then the suffix $\delta_i$ of $\delta$ that starts from $\sigma_{j+1}$ is derivable from $S_i$ using a~left derivation of the context-free grammar~$\mG(\alpha)$.
\item[(d)] If $f(v_i) = u_j$ and $j_{l-1} < j = j_l = j_{l+1} = \ldots = j_{l+m} < j_{l+m+1}$ then there exists $0 \leq n \leq m+1$ such that the suffix $\delta_i$ of $\delta$ that starts from $(\psi_{l+n}?)$ if $n \leq m$ or from $\sigma_{j+1}$ if $n = m+1$ is derivable from $S_i$ using a~left derivation of the context-free grammar $\mG(\alpha)$.
\end{enumerate}

With $\varphi_0 = \lDmd{\alpha}\varphi$ and $f(v_0) = u_0$, the invariants clearly hold for $i = 0$.

Set $i := 0$. While $\varphi_i \neq \varphi$ do:
\begin{itemize}
\item Case $v_i$ is expanded using a~static rule and $\varphi_i$ is the principal formula:
  \begin{itemize}

  \item Case $\varphi_i = \lDmd{\beta;\gamma}\psi$ : Let $v_{i+1}$ be the only successor of $v_i$, $\varphi_{i+1} = \lDmd{\beta}\lDmd{\gamma}\psi$, $f(v_{i+1}) = f(v_i)$, and set $i := i+1$. Clearly, the invariants still hold.

  \item Case $\varphi_i = \lDmd{\xi?}\psi$ : Let $v_{i+1}$ be the only successor of $v_i$, $\varphi_{i+1} = \psi$, and $f(v_{i+1}) = f(v_i)$. Observe that the invariant (a) clearly holds for $i+1$. As $\varphi_i$ is satisfied at $f(v_i)$, both $\xi$ and $\psi$ are satisfied at $f(v_{i+1}) = f(v_i)$. Hence, the label of $v_{i+1}$ is satisfied at $f(v_{i+1})$, and the invariant (b) holds for $i+1$. Let $f(v_i) = u_j$. By the invariants (c) and (d) for $i$, we have that $j \in \{j_1,\ldots,j_{k-1}\}$. As $\delta_i$ is derivable from $S_i = (\xi?)S_{i+1}$ using a~left derivation of the context-free grammar $\mG(\alpha)$, the word $\delta_{i+1}$ such that $\delta_i = (\xi?)\delta_{i+1}$ is derivable from $S_{i+1}$ using a~left derivation of $\mG(\alpha)$. Therefore, by setting $i := i+1$, the invariants (a)-(d) still hold (for the new $i$).

  \item Case $\varphi_i = \lDmd{\beta \cup \gamma}\psi$ : Let $\psi = \triangle_i'\varphi$ and let $S_i'$ be the program sequence corresponding to $\triangle_i'$. By the invariants (c) and (d), $\delta_i$ is derivable from $(\beta \cup \gamma)S_i'$ using a~left derivation of $\mG(\alpha)$. If the first step of that derivation gives $\beta S_i'$ then let $\varphi_{i+1} = \lDmd{\beta}\psi$ else let $\varphi_{i+1} = \lDmd{\gamma}\psi$. By $(\star)$, it follows that $\varphi_{i+1}$ is satisfied at the state $f(v_i)$. Let $v_{i+1}$ be the successor of $v_i$ such that $\varphi_{i+1}$ belongs to the label of $v_{i+1}$. Clearly, the invariant (a) holds for $i+1$. Let $f(v_{i+1}) = f(v_i)$. Thus, the invariants (b)-(d) also hold for $i+1$. Therefore, by setting $i := i+1$, the invariants (a)-(d) still hold (for the new $i$).

  \item Case $\varphi_i = \lDmd{\beta^*}\psi$ : Let $\psi = \triangle_i'\varphi$ and let $S_i'$ be the program sequence corresponding to $\triangle_i'$. By the invariants (c) and (d), $\delta_i$ is derivable from $(\beta^*)S_i'$ using a~left derivation of $\mG(\alpha)$. If the first step of that derivation gives $S_i'$ then let $\varphi_{i+1} = \psi$ else let $\varphi_{i+1} = \lDmd{\beta}\lDmd{\beta^*}\psi$. Let $v_{i+1}$ be the successor of $v_i$ such that $\varphi_{i+1}$ belongs to the label of $v_{i+1}$, let $f(v_{i+1}) = f(v_i)$, and set $i := i+1$. Similarly to the above case, the invariants (a)-(d) still hold.
  \end{itemize}

\item Case $v_i$ is expanded using a~static rule but $\varphi_i$ is not the principal formula: 
  \begin{itemize}
  \item Case the principal formula is not of the form $\lDmd{\alpha'}\varphi'$: Let $v_{i+1}$ be the successor of $v_i$ such that $(v_i,v_{i+1})$ is an edge of $G'$ and the label of $v_{i+1}$ is satisfied at the state $f(v_i)$ of $\mM$. Such a~node $v_{i+1}$ exists because the label of $v_i$ is satisfied at the state $f(v_i)$ of $\mM$. Let $\varphi_{i+1} = \varphi_i$, $f(v_{i+1}) = f(v_i)$, and set $i := i+1$. Clearly, the invariants still hold.
  \item Case the principal formula is of the form $\lDmd{\alpha'}\varphi'$: During a sequence of applications of static rules between two applications of the transitional rule, proceed as for realizing $\lDmd{\alpha'}\varphi'$ in $G'$ (like for the current $\Dmd$-realization of $\lDmd{\alpha}\varphi$ in $G'$ at $v_0$). This decides how to choose $v_{i+1}$ and has effects on terminating the trace (to obtain a $\Dmd$-realization for $\lDmd{\alpha}\varphi$ in $G'$ at $v_0$). We also choose $\varphi_{i+1} = \varphi_i$ and $f(v_{i+1}) = f(v_i)$. By setting $i := i+1$, the invariants still hold (for the new $i$).
  \end{itemize}

\item Case $v_i$ is expanded using the transitional rule: Let $f(v_i) = u_j$. Then, by the invariants (c) and (d), $\varphi_i$ must be of the form $\lDmd{\sigma_{j+1}}\psi$. Let $(v_i,v_{i+1})$ be the edge of $G$ with the label $\varphi_i$. Let $\varphi_{i+1} = \psi$ and $f(v_{i+1}) = u_{j+1}$. Clearly, the invariant (a) holds for $i+1$. By $(\star)$, $\psi$ is satisfied at the state $u_{j+1}$ of $\mM$. By the invariant (b), the other formulas of the label of $v_{i+1}$ are also satisfied at the state $u_{j+1}$ of $\mM$. That is, the invariant (b) holds for $i+1$. It is easy to see that the invariants (c) and (d) remain true after increasing $i$ by 1. So, by setting $i := i+1$, all the invariants (a)-(d) still hold.
\end{itemize}

It remains to show that the loop terminates. 

Observe that any sequence of applications of static rules that contribute to the trace $(v_0,\varphi_0),\ldots,(v_i,\varphi_i)$ of $\lDmd{\alpha}\varphi$ in $G'$ eventually ends because:
\begin{itemize}
\item each formula of the form $\psi \land \xi$, $\psi \lor \xi$, or $[\beta]\psi$ with $\beta \notin \mindices$ may be reduced at most once; 
\item each formula of the form $\lDmd{\beta}\psi$ with $\beta \notin \mindices$ of the label of any node among $v_0,\ldots,v_i$ is reduced according to some $\Dmd$-realization. 
\end{itemize}

Therefore, sooner or later either $\varphi_i = \varphi$ or $v_i$ is a~node that is expanded by the transitional rule. In the second case, if $f(v_i) = u_j$ then $f(v_{i+1}) = u_{j+1}$. As the image of $f$ is $\{u_0,\ldots,u_{j_k}\}$, the construction of the trace must end at some step (with $\varphi_i = \varphi$) and we
obtain a~$\Dmd$-realization in $G'$ for $\lDmd{\alpha}\varphi$ at $v_0$. This completes the proof.\koniec
\end{proof}


\subsection{Model Graphs}

We will prove completeness of \cPDL\ via model graphs. The technique has previously been used in~\cite{Rautenberg83,Gore99,nguyen01B5SL} for logics without the star operator. A {\em model graph} is a~tuple
$\langle W,(R_\sigma)_{\sigma \in \mindices}, H\rangle$, where $W$ is a~set of nodes, $R_\sigma$ for $\sigma \in
\mindices$ is a~binary relation on $W$, and $H$ is a~function that maps each node of $W$ to a~set of formulas. We use
model graphs merely as data structures, but we are interested in ``consistent'' and ``saturated'' model graphs defined
below.

Model graphs differ from ``and-or'' graphs in that a~model graph contains only ``and''-nodes and its edges are labeled by
atomic programs. Roughly speaking, given an ``and-or'' graph $G$ with a~consistent marking $G'$, to construct a~model graph
one can stick together the nodes in a~``saturation path'' of a~node of $G'$ to create a~node for the model graph.
Details will be given later.

A trace of a~formula $\lDmd{\alpha}\varphi$ at a~node in a~model graph is defined analogously as for the case of ``and-or'' graphs. Namely, given a~model graph $\mM = \langle W, (R_\sigma)_{\sigma \in \mindices}, H\rangle$ and a~node $v \in
W$, a~{\em trace} of a~formula $\lDmd{\alpha}\varphi \in H(v)$ (starting from $v$) is a~sequence $(v_0,\varphi_0)$, \ldots, $(v_k,\varphi_k)$ such that:
\begin{itemize}
\item $v_0 = v$ and $\varphi_0 = \lDmd{\alpha}\varphi$;
\item for every $1 \leq i \leq k$, $\varphi_i \in H(v_i)$;
\item for every $1 \leq i \leq k$, if $v_i = v_{i-1}$ then:
  \begin{itemize}
  \item if $\varphi_{i-1} = \lDmd{\beta;\gamma}\psi$ then $\varphi_i = \lDmd{\beta}\lDmd{\gamma}\psi$,
  \item else if $\varphi_{i-1} = \lDmd{\beta \cup \gamma}\psi$ then $\varphi_i = \lDmd{\beta}\psi$ or $\varphi_i = \lDmd{\gamma}\psi$,
  \item else if $\varphi_{i-1} = \lDmd{\beta^*}\psi$ then $\varphi_i = \psi$ or $\varphi_i = \lDmd{\beta}\lDmd{\beta^*}\psi$,
  \item else $\varphi_{i-1}$ is of the form $\lDmd{\psi?}\xi$ and $\varphi_i = \xi$;
  \end{itemize}
\item for every $1 \leq i \leq k$, if $v_i \neq v_{i-1}$ then:
  \begin{itemize}
  \item $\varphi_{i-1}$ is of the form $\lDmd{\sigma}\psi$ and $\varphi_i = \psi$ and $(v_{i-1},v_i) \in R_\sigma$.
  \end{itemize}
\end{itemize}

A trace $(v_0,\varphi_0)$, \ldots, $(v_k,\varphi_k)$ of $\lDmd{\alpha}\varphi$ in a~model graph $\mM$ is called a~{\em
$\Dmd$-realization} for $\lDmd{\alpha}\varphi$ at $v_0$ if $\varphi_k = \varphi$.

Similarly as for markings of ``and-or'' graphs, we define that a~model graph $\mM = \langle W, (R_\sigma)_{\sigma \in
\mindices}, H\rangle$ is {\em consistent} if:
\begin{description}
\item[local consistency:] for every $v \in W$, $H(v)$ contains neither $\bot$ nor a~clashing pair of the form $p$, $\lnot p\,$;
\item[global consistency:] for every $v \in W$, every formula $\lDmd{\alpha}\varphi$ of $H(v)$ has a~$\Dmd$-realization.
\end{description}

A model graph $\mM = \langle W,(R_\sigma)_{\sigma \in \mindices}, H\rangle$ is said to be {\em saturated} if the
following conditions hold for every $v \in W$ and $\varphi \in H(v)$~:
\begin{itemize}
\item if $\varphi = \psi \land \xi$ then $\{\psi,\xi\} \subset H(v)$,
\item if $\varphi = \psi \lor \xi$ then $\psi \in H(v)$ or $\xi \in H(v)$,
\item if $\varphi = \lDmd{\psi?}\xi$ then $\psi \in H(v)$,\footnote{The condition $\xi \in H(v)$ is taken care of by global consistency.}
\item if $\varphi = [\alpha;\beta]\psi$ then $[\alpha][\beta]\psi \in H(v)$,
\item if $\varphi = [\alpha \cup \beta]\psi$ then $\{[\alpha]\psi,[\beta]\psi\} \subset H(v)$,
\item if $\varphi = [\psi?]\xi$ then $\overline{\psi} \in H(v)$ or $\xi \in H(v)$,
\item if $\varphi = [\alpha^*]\psi$ then $\{\psi,[\alpha][\alpha^*]\psi\} \subset H(v)$,
\item if $\varphi = [\sigma]\psi$ and $(v,w) \in R_\sigma$ then $\psi \in H(w)$.
\end{itemize}

Given a~model graph $\mM = \langle W,(R_\sigma)_{\sigma \in \mindices}, H\rangle$, the Kripke model $\mM'$ defined by
$\Delta^{\mM'} = W$, $\sigma^{\mM'} = R_\sigma$ for $\sigma \in \mindices$, and $p^{\mM'} = \{w \in W \mid p \in
H(w)\}$ for $p \in \props$ is called the Kripke model corresponding to~$\mM$.

\begin{lemma} \label{lemma: gen clash}
Let $\mM = \langle W,(R_\sigma)_{\sigma \in \mindices}, H\rangle$ be a~consistent and saturated model graph and let
$\mM'$ be the Kripke model corresponding to~$\mM$. Then, for any $w \in W$, if $\mM',w \models \varphi$ then $H(w)$
does not contain $\overline{\varphi}$.
\end{lemma}

\begin{proof}
By induction on the structure of $\varphi$, using the global consistency.
\end{proof}

\begin{lemma} \label{lemma: model graph}
Let $X$ and $\Gamma$ be finite sets of traditional formulas in NNF and let $\mM = \langle W,(R_\sigma)_{\sigma \in \mindices},
H\rangle$ be a~consistent and saturated model graph such that $\Gamma \subseteq H(w)$ for all $w \in W$, and $X
\subseteq H(\tau)$ for some $\tau \in W$. Then the Kripke model $\mM'$ corresponding to $\mM$ validates $\Gamma$ and
satisfies $X$ at~$\tau$.
\end{lemma}

\begin{proof}
We prove by induction on the construction of $\varphi$ that if $\varphi \in H(w_0)$ for an arbitrary $w_0 \in W$ then
$\mM',w_0 \models \varphi$. It suffices to consider only the non-trivial cases when $\varphi$ is of the form
$\lDmd{\alpha}\psi$ or $[\alpha]\psi$. Suppose that $\varphi \in H(w_0)$.

Consider the case $\varphi = \lDmd{\alpha}\psi$. Let $(w_0,\varphi_0),\ldots,(w_k,\varphi_k)$ be a~$\Dmd$-realization
for $\varphi$ at $w_0$. We have that $\varphi_0 = \varphi$ and $\varphi_k = \psi$. Let $0 \leq i_1 < \ldots < i_h < k$
be all the indices such that, for $1 \leq j \leq h$, $\varphi_{i_j}$ is of the form $\lDmd{\omega_{i_j}}\varphi_{i_j +
1}$ with $\omega_{i_j}$ of the form $\sigma_{i_j}$ or $\psi_{i_j}?$. Observe that
$\omega_{i_1}\omega_{i_2}\ldots\omega_{i_h} \in \mL(\alpha)$ and there is a~path from $w_0$ to $w_k$ in $\mM$ whose
edges are sequently labeled by those $\omega_{i_j}$ of the form $\sigma_{i_j}$. Since $\mM$ is saturated, for $1 \leq j
\leq h$, if $\omega_{i_j} = (\psi_{i_j}?)$ then $\psi_{i_j} \in H(w_{i_j})$, which, by the inductive assumption,
implies that $\mM',w_{i_j} \models \psi_j$. It follows that $(w_0,w_k) \in \alpha^{\mM'}$. Since $\psi \in H(w_k)$, by
the inductive assumption, we have $\mM',w_k \models \psi$. Therefore $\mM',w_0 \models \lDmd{\alpha}\psi$.

Consider the case $\varphi = [\alpha]\psi$. Let $w$ be an arbitrary node of $\mM$ such that $(w_0,w) \in
\alpha^{\mM'}$. We show that $\psi \in H(w)$. There exists a~word $\delta = \omega_1\ldots\omega_k \in \mL(\alpha)$
such that $(w_0,w) \in \delta^{\mM'}$. Let $1 \leq i_1 < \ldots < i_h \leq k$ be all the indices such that, for $1 \leq
j \leq h$, $\omega_{i_j}$ is of the form $\psi_{i_j}?$. For $1 \leq i \leq k$ such that $i \notin \{i_1,\ldots,i_h\}$,
let $\omega_i = \sigma_i$. There exist $w_1,\ldots,w_k \in W$ such that $w_k = w$ and, for $1 \leq i \leq k$, if
$\omega_i$ is $\sigma_i$ then $(w_{i-1},w_i) \in \sigma_i^{\mM'}$, else ($i \in \{i_1,\ldots,i_h\}$ and $\omega_i =
(\psi_i?)$ and) $w_i = w_{i-1}$ and $w_i \in \psi_i^{\mM'}$, which, by Lemma~\ref{lemma: gen clash}, implies that
$\overline{\psi_i} \notin H(w_i)$. Consider the left derivation of $\omega_1\ldots\omega_k$ from $\alpha$ using the
context-free grammar $\mG(\alpha)$. By induction along this derivation, it can be shown that, for $1 \leq i \leq k$,
there exists a~sequence $\triangle_i$ of universal modal operators such that $\triangle_i\psi \in H(w_i)$ and
$\omega_{i+1}\ldots\omega_k$ is derivable from the program sequence corresponding to $\triangle_i$ using a~left
derivation of $\mG(\alpha)$. Hence $\psi \in H(w_k)$, i.e.,\ $\psi \in H(w)$. By the inductive assumption, it follows
that $\mM',w \models \psi$. Therefore $\mM',w_0 \models \varphi$, which completes the proof.
\end{proof}


\subsection{Completeness}

Let $G$ be an ``and-or'' graph for $(X,\Gamma)$ with a~consistent marking $G'$ and let $v$ be a~node of $G'$. A~{\em saturation path} of $v$ w.r.t.\ $G'$ is a finite sequence $v_0 = v$, $v_1$, \ldots, $v_k$ of nodes of $G'$, with $k \geq 0$, such that, for every $0 \leq i < k$, $v_i$ is an ``or''-node and $(v_i,v_{i+1})$ is an edge of $G'$, and $v_k$ is an ``and''-node.

\begin{lemma} \label{lemma: existence of s-p}
Let $G$ be an ``and-or'' graph for $(X,\Gamma)$ with a~consistent marking $G'$. Then each node $v$ of $G'$ has a saturation path w.r.t.\ $G'$. 
\end{lemma}

\begin{proof}
We construct a saturation path $v_0, v_1, \ldots$ of $v$ w.r.t.\ $G'$ as follows. Set $v_0 = v$ and $i = 0$. While $v_i$ is not an ``and''-node do: 
\begin{itemize}
\item If the principal of the static rule expanding $v_i$ is not of the form $\lDmd{\alpha}\varphi$ then let $v_{i+1}$ be any successor of $v_i$ that belongs to $G'$ and set $i := i+1$. 
\item If the principal of the static rule expanding $v_i$ is of the form $\lDmd{\alpha}\varphi$ then:
  \begin{itemize}
  \item let $v_{i+1},\ldots,v_j$ be the longest sequence of ``or''-nodes of $G'$ such that there exist formulas $\varphi_{i+1}$, \ldots, $\varphi_j$ such that the sequence $(v_i,\varphi_i)$, \ldots, $(v_j,\varphi_j)$ is a prefix of a $\Dmd$-realization in $G'$ for $\lDmd{\alpha}\varphi$ at $v_i$; 
  \item set $i := j$.
  \end{itemize}
\end{itemize}

The loop terminates because each formula not of the form $\lDmd{\alpha}\varphi$ may be reduced at most once.\koniec
\end{proof}

\begin{lemma}[Completeness] \label{lemma: comp PDL}
Let $X$ and $\Gamma$ be finite sets of traditional formulas in NNF, and let $G$ be an ``and-or'' graph for $(X,\Gamma)$. Suppose that
$G$ has a~consistent marking $G'$. Then $X$ is satisfiable w.r.t.\ the set $\Gamma$ of global assumptions.
\end{lemma}

\begin{proof}
We construct a~model graph $\mM = \langle W,(R_\sigma)_{\sigma \in \mindices}, H\rangle$ as follows:
\begin{enumerate}
\item
Let $v_0$ be the root of $G'$ and $v_0,\ldots,v_k$ be a~saturation path of $v_0$ w.r.t.\ $G'$. Set $R_\sigma = \emptyset$ for all $\sigma \in \mindices$ and set $W = \{\tau\}$, where $\tau$ is a~new node. Set $H(\tau) := \mL(v_k) \cup \rfs(v_k)$. Mark $\tau$ as {\em unresolved} and set $f(\tau) = v_k$. (Each node of $\mM$ will be marked either as unresolved or as resolved, and $f$ will map each node of $\mM$ to an ``and''-node of $G'$.)

\item
While $W$ contains unresolved nodes, take one unresolved node $w_0$ and do:
 \begin{enumerate}
 \item For every $\lDmd{\sigma}\lDmd{\alpha_1}\ldots\lDmd{\alpha_h}\varphi \in H(w_0)$, where $\varphi$ is not of the form $\lDmd{\beta}\psi$,~do:
  \begin{enumerate}
  \item Let $\varphi_0 = \lDmd{\sigma}\lDmd{\alpha_1}\ldots\lDmd{\alpha_h}\varphi$, $\varphi_i = \lDmd{\alpha_i}\ldots\lDmd{\alpha_h}\varphi$ for $1 \leq i \leq h$, and $\varphi_{h+1} = \varphi$. Let $u_0 = f(w_0)$. (As a~maintained property of $f$, $\varphi_0$ belongs to the label of $u_0$.) Let the sequence $(u_0,\varphi_0)$, $(u_1,\varphi_1)$ be a~$\Dmd$-realization in $G'$ for $\varphi_0$ at $u_0$. Let $i_1 = 1$. For $1 \leq l \leq h$, let the sequence $(u_{i_l},\varphi_l),\ldots,(u_{i_{l+1}},\varphi_{l+1})$ be a~$\Dmd$-realization in $G'$ for $\varphi_l$ at $u_{i_l}$. Let $u_{i_{h+1}},\ldots,u_m$ be a~saturation path of $u_{i_{h+1}}$ w.r.t.~$G'$.

  \item Let $j_0 = 0 < j_1 < \ldots < j_{n-1} < j_n = m$ be all the indices such that, for $0 \leq j \leq m$, $u_j$ is an ``and''-node of $G$ iff $j \in \{j_0,\ldots,j_n\}$. For $0 \leq s \leq n-1$, let $\lDmd{\sigma_s}\psi_s$ be the label of the edge $(u_{j_s},u_{j_s + 1})$ of $G'$. (We have that $\sigma_0 = \sigma$.)

  \item For $1 \leq s \leq n$ do:
     \begin{enumerate}
     \item Let $Z_s = \mL(u_{j_s}) \cup \rfs(u_{j_s})$.
     \item If there does not exist $w_s \in W$ such that $H(w_s) = Z_s$ then: add a~new node $w_s$ to $W$, set $H(w_s) = Z_s$, mark $w_s$ as unresolved, and set $f(w_s) = u_{j_s}$.
     \item Add the pair $(w_{s-1},w_s)$ to $R_{\sigma_{s-1}}$.
     \end{enumerate}
  \end{enumerate}

 \item Mark $w_0$ as resolved.
 \end{enumerate}
\end{enumerate}

As $H$ is a~one-to-one function and $H(w)$ of each $w \in W$ is a~subset of $FL(X \cup \Gamma)$, the above construction terminates and results in a~finite model graph.

Observe that, in the above construction we transform the sequence $u_0,\ldots,u_m$ of nodes of $G'$, which is a~trace of $\varphi_0$ at $u_0$ that ends with $\varphi$ at $u_m$, to a~sequence $w_0,\ldots,w_n$ of nodes of $\mM$ by sticking together nodes in every maximal saturation path and using both the sets $\mL(u_i)$ and $\rfs(u_i)$. Hence, $\mM$ is saturated and satisfies the local and global consistency properties. That is, $\mM$ is a~consistent and saturated model graph.

Consider Step 1 of the construction. As the label of $v_0$ is $X \cup \Gamma$, we have that $X \subseteq H(\tau)$ and $\Gamma \subseteq H(\tau)$. Consider Step 2(a)iii of the construction, as $u_{j_{s-1}}$ is an ``and''-node and $u_{j_{s-1} + 1}$ is a~successor of $u_{j_{s-1}}$ that is created by the transitional rule, the label of $u_{j_{s-1} + 1}$ contains $\Gamma$, and hence the set $\mL(u_{j_s}) \cup \rfs(u_{j_s})$ also contains $\Gamma$. Hence $\Gamma \subseteq H(w_s)$ for every $w_s \in W$. By Lemma~\ref{lemma: model graph}, the Kripke model corresponding to $\mM$ validates $\Gamma$ and satisfies $X$ at $\tau$. Hence, $X$ is satisfiable w.r.t.~$\Gamma$.
\end{proof}


\section{Soundness and Completeness of \cPDLA}
\label{section: proof PDLA}

\begin{lemma}[Soundness] \label{lemma: soundness PDLA}
Let $\mA$ be an ABox and $\Gamma$ be a~TBox such that $\mA$ is satisfiable w.r.t.\ $\Gamma$. Then any ``and-or'' graph for $(\mA,\Gamma)$ has a~consistent marking.
\end{lemma}

\begin{proof}
Let $G$ be an ``and-or'' graph for $(\mA,\Gamma)$ and let $v_0$ be the root of $G$. Clearly, $\mL(v_0)$ is satisfiable w.r.t.\ $\Gamma$. Let $\mM$ be a Kripke model that satisfies $\mL(v_0)$ and validates $\Gamma$.
We first construct a sequence $v_0, \ldots, v_k$ of nodes of $G$ such that:
\begin{enumerate}
\item[(1)] for $1 \leq i \leq k$, $v_i$ is a successor of $v_{i-1}$ and $\mM$ satisfies $\mL(v_i)$; 
\item[(2)] for $0 \leq i < k$, $v_i$ is an ``or''-node; 
\item[(3)] $v_k$ is an ``and''-node; 
\item[(4)] each formula of the form $a\!:\!\lDmd{\alpha}\varphi$ of $\mL(v_k)$ has a static realization at $v_k$. 
\end{enumerate}

Set $i := 0$. While $v_i$ is not an ``and''-node, do: 
\begin{itemize}
\item If the static rule expanding $v_i$ is not one of $\rDmdScP$, $\rDmdCpP$, $\rDmdQmP$, $\rDmdStP$, then let $v_{i+1}$ be any successor of $v_i$ such that $\mM$ satisfies $\mL(v_{i+1})$, and set $i := i+1$.
\item Else: Let the principal formula of rule applied to $v_i$ be $a\!:\!\lDmd{\alpha}\varphi$. 
  \begin{itemize}
  \item If $Next(a\!:\!\lDmd{\alpha}\varphi)$ is not defined (i.e.\ no strategy for reducing $a\!:\!\lDmd{\alpha}\varphi$ has been established): We have that $a^\mM \in (\lDmd{\alpha}\varphi)^\mM$. Thus, there exists a word $\delta \in L(\alpha)$ and a state $u \in \varphi^\mM$ such that $(a^\mM,u) \in \delta^\mM$. Since $\delta \in \mL(\alpha)$, $\delta$ is derivable from $\alpha$ using a left derivation of the context-free grammar~$\mG(\alpha)$. Set $\beta := \alpha$ and set $\triangle$ to the existential modality corresponding to $\beta$. While $\triangle$ is not the empty modality and does not start with a modal operator of the form $\lDmd{\sigma}$ and $Next(a\!:\!\triangle\varphi)$ is not defined: set $\beta$ to the next expression in the mentioned derivation of $\delta$; let $\triangle'$ be the existential modality corresponding to $\beta$; set $Next(a\!:\!\triangle\varphi) := (a\!:\!\triangle'\varphi)$; and set $\triangle := \triangle'$. It is easy to see that $a^\mM \in (\triangle\varphi)^\mM$ is an invariant of this loop. 

  \item Let $Next(a\!:\!\lDmd{\alpha}\varphi) = (a\!:\!\triangle\psi)$. Note that $a\!:\!\triangle\psi$ must belong to the label of one of the successors of $v_i$ as a formula obtained from $a\!:\!\lDmd{\alpha}\varphi$. Let $v_{i+1}$ be such a successor of $v_i$. By the above mentioned invariant, $a^\mM \in (\triangle\psi)^\mM$. Hence, $\mM$ satisfies $\mL(v_{i+1})$. Set $i := i+1$ to continue the main loop.  
  \end{itemize}
\end{itemize}

The loop must terminate because all paths of complex nodes are finite (see the proof of Lemma~\ref{lemma: size constr and-or gr ABox}). Set $k := i$. The sequence $v_0, \ldots, v_k$ clearly satisfies Conditions (1)-(3). We show that it also satisfies Condition (4). Let $a\!:\!\lDmd{\alpha}\varphi \in \mL(v_k)$. We prove that $a\!:\!\lDmd{\alpha}\varphi$ has a static realization at $v_k$. Since $v_k$ is an ``and''-node, either $\alpha$ is an atomic program (and it is done) or $a\!:\!\lDmd{\alpha}\varphi \in \rfs(v_k)$. Consider the second case. There must exist $0 \leq i < k$ such that $a\!:\!\lDmd{\alpha}\varphi$ is the principal formula of the tableau rule applied to $v_i$. The partial function $Next$ determines a static realization for $a\!:\!\lDmd{\alpha}\varphi$ at~$v_k$.  

We construct a consistent marking $G'$ of $G$ as follows. At the beginning, $G'$ contains the nodes $v_0,\ldots,v_k$ and the edges $(v_i,v_{i+1})$ for $0 \leq i < k$. Next, add to $G'$ all successors $v$ of $v_k$, which are simple nodes of $G$, together with the edges $(v_k,v)$. Then, for every simple node $v$ of $G'$ and for every successor $w$ of $v$ in $G$, if $\mM$ satisfies $\mL(w)$ then add the node $w$ and the edge $(v,w)$ to $G'$. 

It is easy to see that $G'$ is a~marking of $G$. Also, $G'$ clearly satisfies the local consistency property. The first condition of the global consistency property of $G'$ holds due to the construction of the sequence $v_0,\ldots,v_k$. For the second condition of the global consistency property of $G'$, we can use the same proof as for Lemma~\ref{lemma: soundness PDL}. Therefore $G'$ is a consistent marking of~$G$. This finishes the proof.\koniec
\end{proof}

The definition of ``saturation path'' remains unchanged for the case with ABoxes. The counterpart of Lemma~\ref{lemma: existence of s-p} about existence of saturation paths for the case with ABoxes also holds. For this, one can use the same argumentation as in the proof of Lemma~\ref{lemma: existence of s-p} together with the fact that all paths of complex nodes are finite. 

\begin{lemma}[Completeness]
Let $\mA$ be an ABox, $\Gamma$ a~TBox, and $G$ an ``and-or'' graph for $(\mA,\Gamma)$. Suppose that $G$ has a~consistent
marking $G'$. Then $\mA$ is satisfiable w.r.t.~$\Gamma$.
\end{lemma}

\begin{proof}
We construct a~model graph $\mM = \langle W,(R_\sigma)_{\sigma \in \mindices}, H\rangle$ as follows:
\begin{enumerate}
\item
Let $v_0$ be the root of $G'$ and $v_0,\ldots,v_k$ be a~saturation path of $v_0$ w.r.t.\ $G'$.
Let $W_0$ to the set of all state variables occurring in $\mA$ and set $W = W_0$. For each $a \in W_0$, let $H(a)$ be the set of all $\varphi$ such that $a:\varphi$ belongs to the label of $v_k$, and mark $a$ as {\em unresolved}. (Each node of $\mM$ will be marked either as unresolved or as resolved.) For each $\sigma \in \mindices$, set $R_\sigma = \{(a,b) \mid \sigma(a,b) \in \mA\}$.

\item
While $W$ contains unresolved nodes, take one unresolved node $w_0$ and do:
 \begin{enumerate}
 \item For every $\lDmd{\sigma}\lDmd{\alpha_1}\ldots\lDmd{\alpha_h}\varphi \in H(w_0)$, where $\varphi$ is not of the form $\lDmd{\beta}\psi$,~do:
  \begin{enumerate}
  \item
     \begin{enumerate}
     \item Let $\varphi_0 = \lDmd{\sigma}\lDmd{\alpha_1}\ldots\lDmd{\alpha_h}\varphi$, $\varphi_i = \lDmd{\alpha_i}\ldots\lDmd{\alpha_h}\varphi$ for $1 \leq i \leq h$, and $\varphi_{h+1} = \varphi$.

     \item If $w_0 \in W_0$ then:
        \begin{itemize}
        \item Let $u_0 = v_k$.
        \item Let $u_1$ be the node of $G'$ such that the edge $(u_0,u_1)$ is labeled by $(w_0 : \varphi_0)$. (Recall that $w_0$ is a~state variable and note that $\varphi_1$ belongs to the label of $u_1$.)
        \end{itemize}

     \item Else:
        \begin{itemize}
        \item Let $u_0 = f(w_0)$. ($f$ is a~constructed mapping that maps each node of $\mM$ not belonging to $W_0$ to an ``and''-node of $G'$. As a~maintained property of $f$, $\varphi_0$ belongs to the label of $u_0$.)
        \item Let $u_1$ be the node of $G'$ such that the edge $(u_0,u_1)$ is labeled by $\varphi_0$. (Note that $\varphi_1$ belongs to the label of $u_1$.)
        \end{itemize}

     \item Let $i_1 = 1$. For $1 \leq l \leq h$, let the sequence $(u_{i_l},\varphi_l)$, \ldots, $(u_{i_{l+1}},\varphi_{l+1})$ be a~$\Dmd$-realization in $G'$ for $\varphi_l$ at $u_{i_l}$. Let $u_{i_{h+1}}$, \ldots, $u_m$ be a~saturation path of $u_{i_{h+1}}$ w.r.t.~$G'$.
     \end{enumerate}

  \item Let $j_0 = 0 < j_1 < \ldots < j_{n-1} < j_n = m$ be all the indices such that, for $0 \leq j \leq m$, $u_j$ is an ``and''-node of $G$ iff $j \in \{j_0,\ldots,j_n\}$. Let $\sigma_0 = \sigma$. For $1 \leq s \leq n-1$, let $\lDmd{\sigma_s}\psi_s$ be the label of the edge $(u_{j_s},u_{j_s + 1})$ of $G'$.

  \item For $1 \leq s \leq n$ do:
     \begin{enumerate}
     \item Let $Z_s = \mL(u_{j_s}) \cup \rfs(u_{j_s})$.
     \item If there does not exist $w_s \in W$ such that $H(w_s) = Z_s$ then: add a~new node $w_s$ to $W$, set $H(w_s) = Z_s$, mark $w_s$ as unresolved, and set $f(w_s) = u_{j_s}$.
     \item Add the pair $(w_{s-1},w_s)$ to $R_{\sigma_{s-1}}$.
     \end{enumerate}
  \end{enumerate}

 \item Mark $w_0$ as resolved.
 \end{enumerate}
\end{enumerate}

Note that the above construction differs from the construction given in the proof of Lemma~\ref{lemma: comp PDL} mainly
by Steps 1 and 2(a)iB.

The above construction terminates and results in a~finite model graph because that: for every $w,w' \in W \setminus W_0$, $w \neq w'$ implies $H(w) \neq H(w')$, and for every $w \in W$, $H(w)$ is a~subset of $FL(X)$, where $X = \Gamma \cup \{ \varphi \mid (a:\varphi) \in \mA \textrm{ for some } a\}$.

Similarly as for the construction given in the proof of Lemma~\ref{lemma: comp PDL}, it can be seen that $\mM$ is
saturated and satisfies the local consistency property. The global consistency condition clearly holds for nodes from
$W \setminus W_0$. For $w \in W_0$ and $\lDmd{\alpha}\varphi \in H(w)$, observe that the formula has a~trace ending at
some node of $W\setminus W_0$, which then continues to form a~$\Dmd$-realization for $\lDmd{\alpha}\varphi$ at~$w$.
Hence, $\mM$ is a~consistent and saturated model graph.

By the definition of ``and-or'' graphs for $(\mA,\Gamma)$ and monotonicity of the ``prime'' static rules of \cPDLA\ except $\rBotzP$ and $\rBotP$: if $(a:\varphi) \in \mA$ then $\varphi \in H(a)$; if $\sigma(a,b) \in \mA$ then $(a,b) \in R_\sigma$; and $\Gamma \subseteq H(a)$ for all $a \in W_0$. We also have that $\Gamma \subseteq H(w)$ for all $w \in W\setminus W_0$. Hence, by Lemma~\ref{lemma: model graph}, the Kripke model corresponding to $\mM$ validates $\Gamma$ and satisfies~$\mA$. Thus $\mA$ is satisfiable w.r.t.~$\Gamma$.
\end{proof}
} 


\end{document}